\documentclass[times]{rncauth}

\usepackage{moreverb}

\usepackage[dvips,colorlinks,bookmarksopen,bookmarksnumbered,citecolor=red,urlcolor=red]{hyperref}
\usepackage{subfigure}
\usepackage{lineno}
\usepackage{color,soul}
\setulcolor{red}
\setstcolor{green}
\sethlcolor{yellow}
\newcommand{\rtn}{\mathbb{R}}

\newcommand{\kf}{\mathcal{K}}
\newcommand{\kl}{\mathcal{KL}}
\newcommand\BibTeX{{\rmfamily B\kern-.05em \textsc{i\kern-.025em b}\kern-.08em
T\kern-.1667em\lower.7ex\hbox{E}\kern-.125emX}}
\def\volumeyear{2010}
\begin{document}

\newtheorem{theorem}{Theorem}
\newtheorem{assumption}{Assumption}
\newtheorem{remark}{Remark}
\newtheorem{definition}{Definition}
\newtheorem{lemma}{Lemma}
\newtheorem{corollary}{Corollary}
\newtheorem{proposition}{Proposition}
\newtheorem{property}{Property}
\runningheads{Z.ZENG,Other}{A demonstration of the \journalabb\
class file}
\makeatletter
\title{Convergence Analysis using the Edge Laplacian: Robust Consensus of Nonlinear Multi-agent Systems via ISS Method }

\author{Zhiwen Zeng\affil{1}, Xiangke Wang\affil{2}\corrauth , Zhiqiang Zheng\affil{3}}

\address{Mechatronics and Automation School of National University of Defense Technology, Changsha, 410073, China}

\corraddr{Mechatronics and Automation School of National University of Defense Technology, Changsha, 410073, China. \\
E-mail: xkwang@nudt.edu.cn
}

\begin{abstract}
This study develops an original and innovative matrix representation with respect to the information flow for networked multi-agent system. To begin with, the general concepts of the edge Laplacian of digraph are proposed with its algebraic properties. Benefit from this novel graph-theoretic tool, we can build a bridge between the consensus problem and the edge agreement problem; we also show that the edge Laplacian sheds a new light on solving the leaderless consensus problem. Based on the edge agreement framework, the technical challenges caused by unknown but bounded disturbances and inherently nonlinear dynamics can be well handled. In particular, we design an integrated procedure for a new robust consensus protocol that is based
on a blend of algebraic graph theory and the newly developed cyclic-small-gain theorem. Besides, to highlight the intricate relationship between the original graph and cyclic-small-gain theorem, the concept of edge-interconnection graph is introduced for the first time. Finally, simulation results are provided to verify the theoretical analysis.
\end{abstract}
\maketitle
\section{Introduction}
\vspace{-2pt}

Recent decades have witnessed tremendous interest in investigating the distributed coordination of multi-agent systems due to broad applications in several areas including formation flight \cite{beard2002coordinated}, coordinated robotics \cite{lin2005necessary}, sensor fusion \cite{olfati2007distributed} and distributed computations \cite{nedic2009distributed}. Recently, the consensus, as a critical research problem, has received increasing amounts of attention. To address such issue, the study of individual dynamics, communication topologies and distributed controls plays an important role.

Graph theory contributes significantly to the analysis and synthesis of networked multi-agent systems because it provides general abstractions for how information is shared between agents in a network. Particularly, the graph Laplacian plays an important role in the convergence analysis and has been explored in several contexts \cite{olfati2004consensus}\cite{ren2005consensus}. Although the graph Laplacian is a convenient method to describe the geometric interconnection of networked agents, another attractive notion--the edge agreement, which has not been explored extensively, deserves additional attention because the edges are adopted \ul{as} natural interpretations of the information flow. Pioneering studies on edge agreement protocol provide new insights into how certain subgraphs, such as spanning trees and cycles, affect the convergence properties and set up a novel systematic framework for analyzing multi-agent systems from the edge perspective \cite{zelazo2011edge}\cite{zelazo2011graph}\cite{zelazo2013performance}. In these literatures, an important matrix representation is introduced, which is referred to as \emph{edge Laplacian}. The edge agreement in \cite{zelazo2011edge} provides a theoretical analysis of the system's performance using both $H_2$ and $H_\infty$ norms, and these results are applied in relative sensing networks, referring to \cite{zelazo2011graph}. Furthermore, based on the properties of the edge Laplacian, \cite{zelazo2013performance} examines how cycles impact the $H_2$ performance and proposes an optimal strategy for designing consensus networks. It is worth noting that, in above-mentioned studies, the edge Laplacian representation is valid for undirected graphs, rather than the more general directed case. Considering the numerous applications of digraphs in multi-agent coordination, extending such a graph tool to the directed case is of great interest.


The information exchange between agents cannot be accurate in practical applications because it is inevitable that types of external noises exist in sensing and transmission\ul{;} and the existence of perturbations always leads to instability and performance deterioration. Therefore, it is significant to investigate their effects on the behavior of multi-agent systems and design a robust consensus protocol to improve disturbance rejection properties. In \cite{bauso2009consensus}, the concept of $\varepsilon$-consensus for an undirected graph is introduced in which the state errors are required to converge in a small region under unknown but bounded external disturbances. In the presence of external disturbances, \cite{wen2012consensus} guarantees a finite $L_2$-gain performance under strongly connected graphs. The global consensus with guaranteed $H_\infty$ performance can be achieved under a strongly connected graph in \cite{li2012global}, where a Lipschitz continuous multi-agent system subject to external disturbances is considered. Motivated by above observations, we consider a more general directed communication graph containing a spanning tree and has unknown but bounded disturbances in the neighbor’s state feedback.

Recently, researchers have increasingly interested in consensus for multi-agent systems with nonlinear dynamics because most of the physical systems are inherently nonlinear in nature. In \cite{das2011cooperative}, an adaptive pinning control method is introduced to study the synchronization of uncertain nonlinear networked systems. The second-order consensus of multi-agent systems with heterogeneous nonlinear dynamics and time-varying delays is investigated in \cite{zhu2010leader}. In \cite{wang2012formation}, a distributed control law for the formation tracking of a multi-agent system that is governed by locally Lipschitz-continuous dynamics under a directed topology is developed. It should be noted that the studies mentioned above are all in the leader-follower setting; however, we will investigate a substantially challenging problem in this paper, i.e., the leaderless consensus problem for multi-agent systems with inherently nonlinear dynamics. To the best of our knowledge, most of existing studies focus on solving the nonlinear leaderless consensus under undirected graphs \cite{jie2012containment}\cite{mei2013distributed}\cite{ren2009distributed}, while a few considered the directed topology. Only a small number of studies use the concepts of generalized algebraic connectivity \cite{yu2010second}, nonsmooth analysis based on the study of Moreau \cite{lin2007state} and the limit-set-based approach \cite{shi2009global} to solve the state agreement problem with nonlinear multi-agent systems under digraphs. However, due to the fact of these methods' extremely complicated characters, it is difficult to promote them in the practice. The edge agreement model allows the development of a universal solution framework to address these problems.

To address the technical challenges caused by the external disturbances and the inherently nonlinear dynamics, the concept of input-to-state stability (ISS), which reveals how external inputs affect the internal stability of nonlinear systems (see \cite{sontag2008input} for a tutorial), is used in this paper. Actually, the concept of ISS is of significance in analyzing large-scale systems, and considerable efforts have been devoted to the interconnected ISS nonlinear systems; particularly, the general cyclic-small-gain theorem of ISS systems developed in \cite{liu2011lyapunov}. Most recently, the nonlinear small-gain design methods, especially the cyclic-small-gain approach, are utilized to design new distributed control strategies for flocking and containment control in \cite{wang2014coordination} as well as to deal with formation control of nonholonomic mobile robots in \cite{liu2013distributed}. In \cite{liu2013outputfeedback}, the authors present a cyclic-small-gain approach to distributed output-feedback control of nonlinear multi-agent systems.


This paper initially extends the concept of the edge Laplacian to digraphs and explores its algebraic properties. To proceed with a seamless integration of graph theory and ISS designs, we present a new robust consensus protocol induced from the edge agreement model to solve the leaderless consensus problem with unknown but bounded disturbances and inherently nonlinear dynamics. To better comprehend the edge agreement mechanism, both strongly connected and quasi-strongly connected situations are considered in this paper. The contributions of this paper depend on three aspects. First, the edge Laplacian of digraph is proposed as well as the edge \ul{adjacency} matrix. Comparing to our previous works \cite{zeng2014nonlinear}, much more details about the algebraic properties of the edge Laplacian are explored. As a matter of fact, the invertibility of the \ul{incidence} matrix and the spectra properties of the edge Laplacian play a central role in the subsequent analysis. Second, we provide a general framework for analyzing the leaderless consensus problem for digraph based on the edge agreement mechanism. Under such framework, the technical challenges caused by the external disturbances and the inherently nonlinear dynamics can be effectively addressed. In particular, we design an integrated procedure for a new robust consensus protocol which is based on a blend of algebraic graph theory and ISS method. Third, following our setup, the edge-interconnection graph is proposed. One of our primarily goal in this paper is to explicitly highlight the insights that the edge-interconnection structure offers in the analysis and synthesis of multi-agent networks. In this direction, we note a reduced order modeling for the edge agreement in terms of the spanning tree, and based on this observation, a two-subsystem interconnection structure is given.   

The organization of this paper is as follows: In Section 2, a brief overview of the basic concepts and results in graph theory are presented as well as the ISS cyclic-small-gain theorem. The edge Laplacian of digraph and other related concepts are proposed in Section 3. The main content on the robust consensus protocol with inherently nonlinear dynamics is elaborated in Section 4. Numerical simulation results are provided in Section 5. The last section presents the conclusions and proposes a number of future research directions.
\vspace{-6pt}

\section{Basic Concepts and Preliminary Results}\label{sec:basis}

In this section, we present a number of basic concepts in graph theory and the ISS cyclic-small-gain theorem.

\subsection{Graph and Matrix}

In this paper, we use $\left|  \cdot  \right|$ and $\left\|  \cdot  \right\|$ to denote the Euclidean norm and 2-norm for vectors and matrices respectively. The notation $\left\|  \cdot  \right\|_\infty$ is used to denote the supremum norm for a function. Let $\mathcal{R}\left( A \right)$ and $\mathcal{N}\left( A \right)$ denote the range space and null space of matrix $A$.

\hl{Let $\mathcal{G} = \left( {\mathcal{V},\mathcal{E},{A}_\mathcal{G}} \right)$ be a \emph{digraph} of order $N$ with a finite nonempty set of nodes $\mathcal{V}  = \left\{ {1,2, \cdots,N} \right\}$, a set of directed edges $\mathcal{E}  \subseteq \mathcal{V} \times \mathcal{V}$ with size $L$ and a adjacency matrix ${A}_\mathcal{G} = \left[ {\alpha_{ij} } \right] \in \mathbb{R}^{N \times N}$, where $\alpha_{ij} = 1$ if and only if $\left( {i, j} \right) \in \varepsilon$ else $\alpha_{ij}  = 0$.}
The degree matrix $\Delta_\mathcal{G} = \left[\Delta_{ij}\right]$ is a diagonal matrix with $\left[\Delta_{ii}\right] = \sum\nolimits_{j = 1}^N {\alpha _{ij},i = 1,2, \cdots,N}$, being the out-degree of node $i$, and the \emph{graph Laplacian} of the digraph $\mathcal{G}$ is defined by $L_\mathcal{G}=\Delta_\mathcal{G}-{A}_\mathcal{G}$. For a given parent node $i$, its incident edge is called the \emph{parent edge} denoted by $\mathcal{P}^i$; its emergent edge is called the \emph{child edge} denoted by $\mathcal{C}^i$; and the edges derived from the same parent $i$ (\ul{i.e.,} the collection of $\mathcal{C}^i$) are called the \emph{sibling edges} denoted by $\mathcal{E}^i_\mathcal{C}$. We call an edge the neighbor of $e_k$ if they share a node, and the neighbor set of $e_k$ is denoted by $N_{e_k} =  \left\{ e_l\in\mathcal{E}:\text{if $e_k$ and $e_l$ share a node} \right\}$. The  \emph{outgoing neighbors} of $e_k$ refer to  $N_{e_k }^ \otimes = \left\{ {e_l  \in N_{e_k}: \text{ if $e_k$ = $\mathcal{P}^i$ and $e_l = \mathcal{C}^i$; or $e_k, e_l \in  \mathcal{E}^i_\mathcal{C}$, where $i$ is the conjunct node} } \right\}$ (see, e.g., Figure \ref{fig:edge}).

\begin{figure}[hbtp]
\begin{center}
\mbox{
\subfigure[$e_k$ = $\mathcal{P}^i$ and $e_l = \mathcal{C}^i$]
{\includegraphics[height=20 mm]{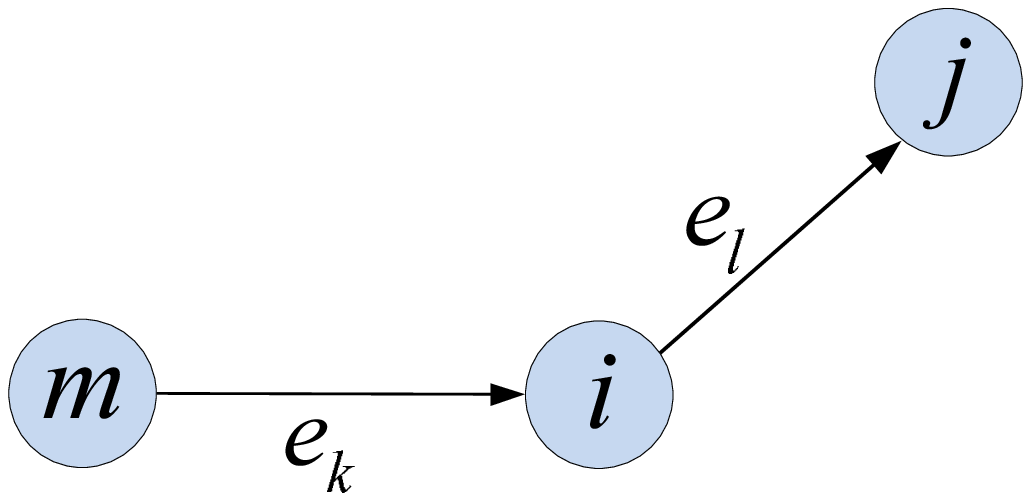}} \hspace{1.0in}
\subfigure[$e_k, e_l \in  \mathcal{E}^i_\mathcal{C}$]
{\includegraphics[height=18 mm]{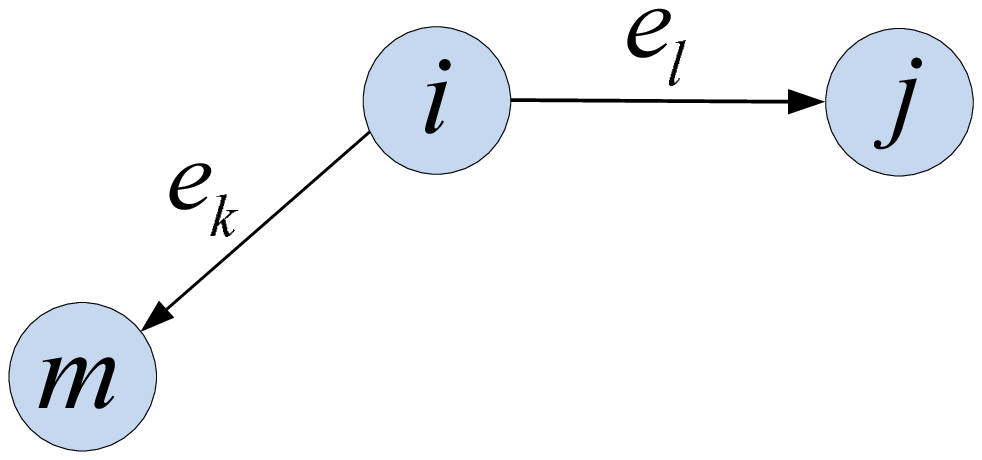}}\quad
}
\caption{ Two typical forms of the outgoing neighbors, $N_{e_k }^ \otimes$.}
\label{fig:edge}
\end{center}
\end{figure}
A directed path in digraph $\mathcal{G}$ is a sequence of directed edges. A directed tree is a digraph in which, for the root $i$ and any other node $j$, there exists exactly one directed path from $i$ to $j$. A spanning tree of a digraph is a directed tree formed by graph edges that connect all the nodes of the graph \cite{godsil2001algebraic}. Graph $\mathcal{G}$ is called strongly connected if and only if any two distinct nodes can be connected via a directed path and quasi-strongly connected if and only if it has a directed spanning tree \cite{thulasiraman2011graphs}.

The incidence matrix $E\left(\mathcal{G} \right)$ for a digraph is a $\left\{ {0, \pm 1} \right\}$-matrix with rows and columns indexed by the vertexes and edges of $\mathcal{G}$, respectively, such that
$$
\left[ {E\left( \mathcal{G} \right)} \right]_{ik} =
\begin{cases}
+1 & \text{if $i$ is the initial node of edge $e_k$} \\
-1 & \text{if $i$ is the terminal node of edge $e_k$}\\
0  & \text{otherwise}
\end{cases}
$$
which implies each column of $E$ contains exactly two nonzero entries ``+1'' and ``-1''. In addition, for a quasi-strongly connected digraph, the rank of the incidence matrix is  $rank(E\left(\mathcal{G} \right)) = N-1$ from \cite{thulasiraman2011graphs}. Figure \ref{figure:exam} depicts an example with its incidence matrix.



\begin{figure}[hbtp]
\centering
{\includegraphics[height=3.0 cm]{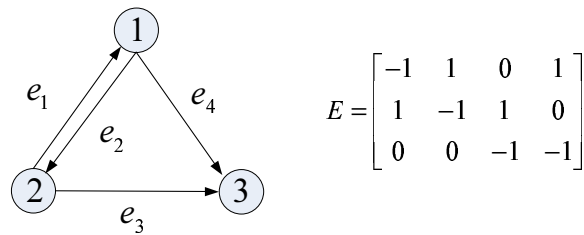}}
\caption{The incidence matrix of a simple digraph.}
\label{figure:exam}
\end{figure}

\subsection{ISS and Small-gain Theorem}

A function $\alpha:\rtn_{+}\to\rtn_+$ is said to be \emph{positive definite} if it is continuous, $\alpha(0)=0$ and $\alpha(s)>0$ for $s>0$. A function $\gamma :{\rtn_ + } \to {\rtn_ + }$ is of class $\kf$ if it is continuous, strictly increasing and $\gamma \left( 0 \right) = 0$; it is of class $\kf_\infty$ if, in addition, it is unbounded. A function $\beta :{\rtn_ + } \times {\rtn_ + } \to {\rtn_ + }$ is of class $\kl$ if, for each fixed $t$, the function $\beta \left( {\cdot,t} \right)$ is of class $\kf$ and for each fixed $s$, the function $\beta \left( {s,\cdot} \right)$ is decreasing and tends to zero at infinity. $\mathrm{Id}$ represents identify function, and symbol $\circ$ denotes the composition between functions.

Consider the following nonlinear system with $x\in\rtn^n$ as the
state and $w\in\rtn^m$ as the external input:
\begin{align} \label{generalizedcontrolledsystem}
     \dot{x}=\zeta(x,w)
\end{align}
where $\zeta:\rtn^n\times\rtn^m\to\rtn^n$ is a locally Lipschitz
vector field.

\begin{definition}[\cite{sontag2008input}] 
The system~\eqref{generalizedcontrolledsystem} is said to be
\emph{input-to-state stable (ISS)} with $w$ as an input if there exist
$\beta\in\kl$ and $\gamma\in\kf$ such that for each initial
condition $x(0)$ and each measurable essentially bounded input
$w(\cdot)$ defined on $[0,\infty)$, the solution $x(\cdot)$ exists
on $[0,\infty)$ and satisfies
\begin{align*}
   |x(t)|\le\beta(|x(0)|,\gamma(\lVert w\rVert_\infty )),~~~\forall t\ge0.
\end{align*}
\end{definition}
It is known that if the system $\dot{x}=f(x,w)$ in
\eqref{generalizedcontrolledsystem} is ISS with $w$ as the input,
then the unforced system $\dot{x}=f(x,0)$ is globally asymptotically
stable at $x=0$.

%

\begin{lemma}[\cite{sontag2008input}] \label{ISSandLyapulov}
The system~\eqref{generalizedcontrolledsystem} is ISS if
and only if it has an ISS-Lyapunov function.
\end{lemma}

Consider the following interconnected system composed of $N$
interacting subsystems:
\begin{align} \label{connectedsystem}
 \dot {x}_i=\zeta_i(x,w_i),~~~i=1,\ldots,N
 \end{align}
where $x_i\in\rtn^{n_i}$, $w_i\in\rtn^{m_i}$ and
$\zeta_i:\rtn^{n+m_i}\to\rtn^{n_i}$ with $n=\sum^N_{i=1} n_i$ is locally
Lipschitz continuous such that $x=[x^T_1,\ldots,x^T_N]^T$ is the
unique solution of system~\eqref{connectedsystem} for a given
initial condition. The external input $w=[w^T_1,\ldots,w^T_N]^T$ is
a measurable and locally essentially bounded function from $\rtn_+$
to $\rtn^{m}$ with $m=\sum^N_{i=1} m_i$. Furthermore, \hl{we use $\gamma _y^x \in \kf$ to represent the gain function from $x$-subsystem to $y$-subsystem} \cite{isidori1999nonlinear}.

\begin{lemma}[Cyclic-small-gain Theorem, \cite{liu2011lyapunov}]\label{lem:CSG}
\label{lem:ISScicle} Consider the continuous-time dynamical
network~\eqref{connectedsystem}. Suppose that for $i=1,\ldots,N$,
the $x_i$-subsystem admits an ISS-Lyapunov function $V_i: \rtn^{n_i}\to\rtn$  satisfying
\begin{itemize}
\item  there exist $\underline{a}_i, \overline{a}_i \in \kf_{\infty}$ such that
\begin{align*}
\underline{a}_i(|x_i|)\le
V_i(x_i)\le\overline{a}_i(|x_i|),~~~\forall x_i;
\end{align*}

\item there exist $\gamma_{x_i}^{x_j}\in\kf\cup\{0\}(j\neq i)$, $\gamma_{x_i}^{w_i}\in\kf\cup\{0\}$ and a positive definite $\alpha_i$ such
that
\begin{align*}
& V_i(x_i)\ge\max\{\gamma_{x_i}^{x_j}(V_j(x_j)),\gamma_{x_i}^{w_i}(|w_i|)\} \nonumber\\
\Rightarrow &
\nabla V_i(x_i)\zeta_i(x,w_i) \le-\alpha_i(V_i(x_i)),\quad \forall x,~\forall w_i.
\end{align*}
\end{itemize}
Then, the system~\eqref{connectedsystem} is ISS if for each $r=2,\ldots
N$,
\begin{align} \label{align:cyclic-small-gain}
\gamma_{x_{i_1}}^{x_{i_2}}\circ\gamma_{x_{i_2}}^{x_{i_3}}\circ\cdots\circ\gamma_{x_{i_r}}^{x_{i_1}}<\mathrm{Id}
\end{align}
for all $1\le i_j\le N, i_j\neq i_{j'}$ if ${j\neq j'}$.
\end{lemma}

\section{The Edge Laplacian of Digraph}\label{sec:Edge Laplacian For Digraph}

The edge Laplacian is a promising graph-theoretic tool, however it still remains to an undirected notion and is thus inadequate to handle our problem. Undoubtedly, extending the concept of the edge Laplacian to the digraph and exploring its algebraic properties will contribute significantly to the investigation of multi-agent systems. In this section, we first give the definition of the in-incidence matrix and out-incidence matrix.

\begin{remark}
As known that the ``out-degree'' relates how each node in the network impacts on other nodes, but the ``in-degree'' directly captures how the dynamics of an agent is influenced by others \cite{mesbahi2010graph}. In fact, the following investigations, including the definitions and properties, all adopt the ``out-degree'' description. However, for the ``in-degree'' case, analogous methods can also be applied.
\end{remark}

\begin{definition}[\emph{In-incidence} and \emph{Out-incidence} Matrix]

The $N \times L$ in-incidence matrix ${E_{\odot} \left( \mathcal{G} \right)}$ for a digraph $\mathcal{G}$ is a $\{ 0, - 1\}$ matrix with rows and columns indexed by nodes and edges of $\mathcal{G}$, respectively, such that
$$
\left[ {E_{\odot} \left( \mathcal{G} \right)} \right]_{ik}  := \begin{cases}
-1 & \text{if $i$ is the terminal node of edge $e_k$} \\
0  & \text{otherwise}
\end{cases}
$$
and the out-incidence matrix is a $\{ 0, + 1\}$ matrix, defined as
$$
\left[ {E_\otimes \left(\mathcal{G} \right)} \right]_{ik}  := \begin{cases}
+1 & \text{if $i$ is the initial node of edge $e_k$} \\
0  & \text{otherwise.}
\end{cases}
$$
\end{definition}

Comparing with the definition of the incidence matrix, we can write $E\left(\mathcal{G} \right)$ in the following manner
\begin{align}\label{align:incidentrelation}
E(\mathcal{G}) = E_ \odot(\mathcal{G})   + E_ \otimes(\mathcal{G}).
\end{align}
In the following discussions, we use $E, E_ \odot$ and $E_ \otimes$ instead of $E(\mathcal{G}), E_ \odot(\mathcal{G})$ and $E_ \otimes(\mathcal{G})$.

Due to the fact that each row of the out-incidence matrix actually can be viewed as a decomposition of the out-degree from a node to each specific edge, we can derive a novel factorization of the graph Laplacian.

\begin{lemma}\label{lem:lap}
Considering a digraph $\mathcal{G}$ with the incidence matrix $E$ and the out-incidence matrix $E_\otimes$, then the graph Laplacian of $\mathcal{G}$ have the following expression:
\begin{equation}\label{equ:Lap}
L_\mathcal{G} = E_\otimes E^T.
\end{equation}
\end{lemma}

\begin{proof}
According to the preceding definition of $E_\otimes$, we obtain that
\begin{equation*}
E_{\otimes i} E_{\otimes j}^T  =  \begin{cases}
\left[\Delta_{ii}\right] & \text{if $i = j$} \\
0  & \text{otherwise}\\
\end{cases}
\end{equation*}
which implies $E_\otimes E_\otimes^T =  \Delta_\mathcal{G}$. Similarly, we have
\begin{equation*}
E_{\otimes i} E_{\odot j}^T  =  \begin{cases}
-1 & \text{if $j \in N_i$} \\
0  & \text{otherwise}\\
\end{cases}
\end{equation*}
where $N_i$ denotes the neighbor set of node $i$, and we can collect \ul{the} terms as $E_\otimes E_\odot^T = - {A}_\mathcal{G}$. According to the definition of the graph Laplacian and equation \eqref{align:incidentrelation}, we have
$$
E_\otimes E^T =  E_\otimes E_\otimes^T+ E_\otimes E_\odot^T=\Delta_\mathcal{G}-{A}_\mathcal{G}=L_\mathcal{G}.
$$
The proof is concluded.
\end{proof}

Next, we give the definition of the edge variant of the graph Laplacian.

\begin{definition}[\emph{Edge Laplacian of Digraph}]
The edge Laplacian of digraph is defined as
\begin{align*}
L_e := E^T E_\otimes
\end{align*}
with $L \times L$ elements.
\end{definition}

\begin{definition}[\emph{Edge Adjacency Matrix}]
$$
\left[ {A_e } \right]_{kl} : =
\begin{cases}
+1 & \text{ $e_k, e_l \in  \mathcal{E}^i_\mathcal{C}$} \\
-1 & \text{ $e_k$ = $\mathcal{P}^i$ and $e_l = \mathcal{C}^i$} \\
0  & \text{otherwise.}
\end{cases}
$$
\end{definition}

For instance, the edge Laplacian matrix and the edge adjacency matrix of the simple digraph shown in Figure \ref{figure:exam} are, respectively
\begin{align*}
L_e =  \left( {\begin{matrix}
   {1} & -1 & 1 & -1   \cr
   -1 & 1 & {-1} & 1   \cr
   1 & { 0} & 1 &  0    \cr
   0 &  1   & 0 &  1       \cr
 \end{matrix} } \right)  ~~\text{and}~~
A_e = \left( {\begin{matrix}
   {0} & -1 & 1 & -1   \cr
   -1 & 0 & {-1} & 1   \cr
   1 & { 0} & 0 &  0    \cr
   0 &  1   & 0 &  0       \cr
 \end{matrix} } \right).
\end{align*}

\begin{lemma}
The edge Laplacian can be constructed from the edge adjacency matrix
\begin{equation}\label{equation:Leeuqal2IaddAe}
L_e  = I + A_e.
\end{equation}
\end{lemma}

\begin{proof}
According to the definition of $E_\otimes$ and $A_e$, the result comes after the proof of Lemma \ref{lem:lap}.
\end{proof}

To provide a deeper insight into what the edge Laplacian $L_e$ offers in the analysis and synthesis of multi-agent systems, we propose the following lemma.

\begin{lemma}\label{theorem:Laplacianeigequal}
For any digraph $\mathcal{G}$, the graph Laplacian $L_{\mathcal{G}}$ and the edge Laplacian $L_e$ have the same nonzero eigenvalues. In addition, the edge Laplacian $L_e$ contains exactly $N-1$ nonzero eigenvalues and all in the open right-half plane, if $\mathcal{G}$ is quasi-strongly connected.
\end{lemma}

\begin{proof}
Suppose that $\lambda \ne 0$ is an eigenvalue of $L_\mathcal{G}$, which is associated with a nonzero eigenvector $p$. Therefore, we have
\begin{align}\label{align:eigen}
L_\mathcal{G} p = E_\otimes E^T p = \lambda p
\end{align}
which implies $E^T  p = \bar p \ne  0$. By left-multiplying both sides of (\ref{align:eigen}) by $E^T$, one can obtain
\begin{align*}
L_e \bar p = E^T E_\otimes E^T p = \lambda \bar p
\end{align*}
which shows that $L_e$ contains the nonzero eigenvalues that $L_{\mathcal{G}}$ has.

By using similar approaches, we can proof that $L_{\mathcal{G}}$ also has all the nonzero eigenvalues of $L_e$. It turns out that the nonzero eigenvalues of $L_\mathcal{G}$ and $L_e$ are identical.

By Lemma 3.3 in \cite{ren2005consensus}, for a quasi-strongly connected digraph $\mathcal{G}$ of order $N$, $L_\mathcal{G}$ has $N-1$ nonzero eigenvalues and all in the open right-half plane, therefore $L_e$ contains exactly $N-1$ nonzero eigenvalues as well. Then we come to the conclusion.
\end{proof}


\begin{lemma}\label{thm:zeroeigen}
Consider a quasi-strongly connected digraph $\mathcal{G}$ of order $N$, the edge Laplacian $L_e$ has $L-N+1$ zero eigenvalues and zero is a simple root of the minimal polynomial of $L_e$.
\end{lemma}

\begin{proof}
\hl{Consider the quasi-strongly connected digraph $\mathcal{G}$, $L_e$ has exactly $N-1$ nonzero eigenvalues; therefore}
\begin{equation}\label{l1}
rank(L_e) \ge N-1
\end{equation}
\hl{and the algebraic multiplicity of the zero eigenvalue of $L_e$ is $L-N+1$. Besides, recall the fact that $L_e =  E^T E_\otimes$ and $rank(E)= N-1$, so we have}
\begin{equation}\label{l2}
rank(L_e) \le rank(E^T) = N-1.
\end{equation}
\hl{Then by combining} \eqref{l1} and \eqref{l2}, \hl{one can obtain $rank(L_e) = N-1$, i.e., the dimension of the null space of $L_e$ is $dim~\mathcal{N}{(L_e)} = L-N+1$. In other words, the geometric multiplicity of zero eigenvalue is $L-N+1$. Clearly, the geometric multiplicity and the algebraic multiplicity of zero eigenvalue are equal, which implies that the corresponding Jordan block for each zero eigenvalue is size one from} \cite{meyer2000matrix}. \hl{That is, zero is a simple root of the minimal polynomial of $L_e$}.
\end{proof}

From \cite{zelazo2007agreement}, for the undirected graph, the dynamics of the edge agreement model can be captured by the reduced order system based on the spanning tree. Next, we will further discuss the similar results under digraph.

Clearly, if the digraph $\mathcal{G}$ is quasi-strongly connected, it can be rewritten as a union form: $\mathcal{G} = \mathcal{G}_\mathcal{T}  \cup \mathcal{G}_\mathcal{C}$, where $\mathcal{G}_\mathcal{T}$ is a given spanning tree and $\mathcal{G}_\mathcal{C}$ is the cospanning tree respectively. Correspondingly, the incidence matrix can be rewritten as
\begin{align*}
 E = \left[ {\begin{matrix}
   { E_\mathcal{T} } & { E_\mathcal{C} }  \cr
 \end{matrix}} \right]
\end{align*}
through some permutations, and $ E_\mathcal{T}, E_\mathcal{C}$ are incidence matrices with respect to $\mathcal{G}_\mathcal{T}$ and $\mathcal{G}_\mathcal{C}$. Similarly, the out-incidence matrix can be rewritten as
\begin{align*}
 E_\otimes = \left[ {\begin{matrix}
   { E_{\otimes \mathcal{T}} } & { E_{\otimes \mathcal{C}}}  \cr
 \end{matrix}} \right].
 \end{align*}
It should be mentioned that $\mathcal{G}_\mathcal{C}$ can be reconstructed from $\mathcal{G}_\mathcal{T}$, which indicates that $ E_\mathcal{T}$ has full column rank \cite{zelazo2007agreement}.

According to the partition, one can represent the edge Laplacian in terms of the block form of the incidence matrix as
\begin{align*}
{ L_e}  = { E^T}  E_\otimes  
 = \left[ {\begin{matrix}
   { E_\mathcal{T}^T{ E_{\otimes \mathcal{T}}}} & { E_\mathcal{T}^T{ E_{\otimes \mathcal{C}}}}  \cr
   { E_\mathcal{C}^T{ E_{\otimes \mathcal{T}}}} & { E_\mathcal{C}^T{ E_{\otimes \mathcal{C}}}}  \cr
 \end{matrix} } \right] 
 =\left[ {\begin{matrix}
   {{ L_{e1}}} & {{ L_{e2}}}  \cr
   {{ L_{e3}}} & {{ L_{e4}}}  \cr
\end{matrix} } \right].
\end{align*}
Also it is useful to express the edge adjacency matrix $ A_e$ as the block representation
\begin{align*}
{ A_e} = \left[ {\begin{matrix}
   {{ A_{e1}}} & {{ A_{e2}}}  \cr
   {{ A_{e3}}} & {{ A_{e4}}}  \cr
 \end{matrix} } \right].
 \end{align*}
Following from (\ref{equation:Leeuqal2IaddAe}), we have
\begin{align*}
 { L_{e1}} = {I_1 +  A_{e1}},~
 { L_{e2}} = { A_{e2}},
~ { L_{e3}} = { A_{e3}}
,~\text{and}~{ L_{e4}} = {I_2 +  A_{e4}}.
\end{align*}

\begin{lemma}\label{lemma:pseudo-inverse}
Considering a quasi-strongly connected digraph $\mathcal{G} = \mathcal{G}_\mathcal{T}  \cup \mathcal{G}_\mathcal{C}$, the pseudoinverse of the incidence matrix  $ E^ \dag$ exists, and there exists a matrix $R$ such that
\begin{equation}
 E^ \dag   = {R}^\dag  E_\mathcal{T}^\dag
\end{equation}
with ${R}^\dag =  R ^T\left[ {R R^T  } \right]^{ - 1}$ and $ E_\mathcal{T}^\dag = \left[ { E_\mathcal{T} ^T  E_\mathcal{T} } \right]^{ - 1}  E_\mathcal{T} ^T $, where ${R}^\dag$ is the right-inverse of $R$ and $ E_\mathcal{T}^\dag$ is the left-inverse of $ E_\mathcal{T}$.

\end{lemma}

\begin{proof}
\hl{Since $ E_\mathcal{T}$ has full column rank, so its left-inverse exists and can be directly obtained by $E_\mathcal{T}^\dag = \left[ { E_\mathcal{T} ^T  E_\mathcal{T} } \right]^{ - 1}  E_\mathcal{T} ^T$. Because the columns of $ E_\mathcal{C}$ are linearly dependent on the columns of $ E_\mathcal{T}$, we have}
\begin{align}\label{equation:xc-xt}
 E_\mathcal{T} T = E_\mathcal{C}
\end{align}
\hl{and then the matrix $T$ can be obtained by}
\begin{align}\label{euqation:T}
T   =  E_\mathcal{T} ^\dag  E_\mathcal{C}.
\end{align}
\hl{The matrix $R$ is now defined as} $R  = \left[ {\begin{matrix}
   I & {T }  \cr
 \end{matrix} } \right]$. \hl{In fact, the rows of the matrix $R$ form a basis for the \emph{cut space} of $\mathcal{G}$} \cite{godsil2001algebraic}. \hl{Besides, the incidence matrix of $\mathcal{G}$ can be written as $E =  E_\mathcal{T} R$. Clearly, $ E_\mathcal{T}$ is of full column rank and $R$ is of full row rank; therefore, the pseudoinverse of $E$ can be calculated by  $E^ \dag   = {R}^\dag  E_\mathcal{T}^\dag$ from} \cite{ben2003generalized}.
\hl{Then we reach the conclusion.}
\end{proof}

\section{ROBUST CONSENSUS OF NONLINEAR MULTI-AGENT SYSTEMS VIA ISS DESIGN}

In this section, a new consensus protocol is presented through a seamless integration of graph theory and ISS design. The newly developed cyclic-small-gain theorem is employed to address the challenges caused by unknown but bounded disturbances and the inherently nonlinear dynamics. Contrary to the well-studied graph Laplacian dynamics, we will analyze and synthesize multi-agent systems with the edge perspective by using edge agreement framework. To facilitate a  better understanding, both strongly connected and quasi-strongly connected situations are considered.

To begin our analysis, the dynamics of the $i$-th agent is defined as
\begin{equation}\label{node:dynamics}
\dot x_i \left( t \right) = f \left( {t,x_i } \right) + w_i\left(t\right)+\mu _i\left(t\right),i = 1,2,\cdots,N
\end{equation}
where $x_i \left( t \right) \in {\rtn}^n$ refers to the state vector of the $i$-th node; $f \left( {t,x_i } \right) :\rtn \times{\rtn}^n \to {\rtn}^n$ denotes a Lipschitz continuous function; $w_i\left(t\right)\in {\rtn}^n$ describes unknown but bounded disturbances with the upper bound $\xi \ge 0$, i.e., $\left| w_i \left( t \right)\right| \le  \xi$; $\mu _i\left(t\right)  \in {\rtn}^n$ represents the control input.

In the presence of the disturbances, we should not expect that agents can accurately reach consensus. Therefore, we introduce the \emph{robust consensus} to describe the influence of the disturbances on the behavior of the system.

\begin{definition}[\emph{Robust Consensus}]
In the presence of disturbances $w(t)$ with the upper bound $\xi$, we design a distributed control law $\mu _i\left(t\right)$ for $i = 1,2,\cdots,N$ such that the agent state $x_i$ governed by \eqref{node:dynamics} can reach robust consensus in the nonlinear-gain sense that
\begin{align*}
\left| {x_i(t)  - x_j(t) } \right| \le \ell\left( {\left\| w(t) \right\|_\infty  } \right),~~\ell \in \kf ~~~\text{as}~~~ t \to \infty.
\end{align*}

\end{definition}

Before proceeding, we make the following assumption.
\begin{assumption}
For the nonlinear function $f\left( {t,x_i } \right)$ in (\ref{node:dynamics}), there exists a nonnegative constant $\eta$ such that
\begin{align*}
\left| {f\left( {t,x_1 } \right) - f\left( {t,x_2 } \right)} \right| \le \eta \left| {x_1  - x_2 } \right|,\forall x_1 ,x_2  \in {\rtn}^n ;t \ge 0.
\end{align*}
\end{assumption}

To reach consensus, we employ the following distributed consensus protocol:
\begin{equation}\label{protocolA1}
\mu _i  \left( t \right) =  - \sum\limits_{j \in N_i } { \left( {x_{i} \left( t \right)  - x_{j} \left( t \right) } \right) + u_i \left( t \right) }, i = 1,2,\cdots,N
\end{equation}
where $u_i \left( t \right)$ is the auxiliary control input yet to be designed.


\subsection{Edge Agreement}

Considering an edge $e_k$, we define the \emph{edge state} as $\tilde x_{e_k}\left( t \right)$, which represents the difference between two agents associated with $e_k$. We then have
\begin{equation}\label{align:edgerepresent}
\tilde x_{e_k}\left( t \right) = x_{\otimes \left( e_k \right)} \left( t \right)- x_{\odot \left( e_k \right)}\left( t \right)
\end{equation}
where $\otimes \left( e_k \right)$ and $\odot \left( e_k \right)$ denote the initial node and the terminal node of $e_k$, respectively.

Following this, we obtain
\begin{equation}\label{equation:edgeincidence}
\tilde x_e \left( t \right) = E^T x\left( t \right)
\end{equation}
where $\tilde x_e \left( t \right)$ is the collection of $\tilde x_{e_k}\left( t \right)$. Consider the well-known graph Laplacian dynamics (linear and noise-free) in \cite{olfati2004consensus} as
\begin{align}\label{Lapdyn}
\dot x\left( t \right) =  - Lx\left( t \right).
\end{align}
Differentiating \eqref{equation:edgeincidence} and substituting in \eqref{Lapdyn}, leads to
\begin{align}\label{ELapdyn}
 {\dot {\tilde x}_e}\left( t \right) =  - {L_e}{\tilde x_e}\left( t \right)
\end{align}
which is referred to the \emph{edge agreement protocol}. In comparison to the consensus problem, the edge agreement, rather than requiring the convergence to the agreement subspace \cite{olfati2004consensus}, expects the edge dynamics \eqref{ELapdyn} to converge to the origin. Essentially, the evolution of an edge state depends on its current state and the states of its adjacency edges. In addition, the edge agreement of $\tilde x_e$ implies consensus if the digraph $\mathcal{G}$ has a spanning tree \cite{zelazo2011edge}.



\begin{remark}
Recall that the objective leaderless consensus is defined as $\mathop {\lim }\limits_{t \to \infty} \left| {x_i \left( t \right) - x_j\left( t \right)} \right|=0$. Obviously, directly \ul{modeling} such problems is difficult since the coupling of different agents $x_i \left( t \right)$ and $x_j \left( t \right)$ is involved. However, the edge agreement framework provides a possible method of studying the leaderless consensus problem from the edge perspective. By using \eqref{align:edgerepresent} and \eqref{equation:edgeincidence}, we can turn the consensus problem into an edge agreement problem, where the asymptotic stability of $\tilde x_e$ implies the consensus. In fact, based on such framework, the leaderless consensus problem can be extremely simplified. Moreover, we also suggest that the edge agreement does not only provide a \ul{broader} scope for addressing the leaderless consensus problem but also has a significant potential to address the leader-follower case.
\end{remark}

Given the protocol (\ref{protocolA1}), the graph Laplacian dynamics is obtained as
\begin{equation}\label{vertexdynamics}
\dot x \left( t \right) = \mathcal{F}\left( {t,x} \right) - L_\mathcal{G} x \left( t \right) + w \left( t \right) + u \left( t \right)
\end{equation}
where $\mathcal{F}\left( {t,x}\right)$, $w \left( t\right)$ and $\mu \left( t\right)$ are the column stack vectors of $f \left( {t,x_i }\right)$, $w_i \left( t \right)$ and $\mu _i \left( t \right)$, for $i = 1,2,\cdots,N$.

By differentiating (\ref{equation:edgeincidence}) and substituting into (\ref{vertexdynamics}), we have the following edge Laplacian dynamics:
\begin{align}\label{edge:subsystem}
\dot {\tilde x}_e \left( t \right)
 = E^T \mathcal{F}\left( {t,x} \right) - L_e {\tilde x_e} \left( t \right) +E^T w \left( t \right) + E^T u\left( t \right)
\end{align}
where $u\left( t \right)$ needs to be further determined. For convenience, we define
\begin{equation}\label{viutrual:controllaw}
u_e \left( t \right)  = E^T u \left( t \right).
\end{equation}
From \eqref{equation:Leeuqal2IaddAe}, we have
\begin{align*}
L_{e_k } =  1 + \sum\limits_{e_l  \in N_{e_k }^ \otimes } { [A_e ]_{kl} }.
\end{align*}
Finally, the sate of the $k$-th edge evolves according to the following system:
\begin{align}\label{edge:finaldynamics}
\dot {\tilde x}_{e_k } \left( t \right) = f \left( {t,x_{\otimes \left( e_k \right) } } \right) - f \left( {t,x_{\odot \left( e_k \right) } } \right) - \left( \tilde x_{e_k }\left( t \right)  +
\sum\limits_{e_l  \in N_{e_k }^ \otimes } { [A_e ]_{kl} } \tilde x_{e_l }\left( t \right)\right) + \left[E^T\right]_k w \left( t \right) + u_{e_k }\left( t \right).
\end{align}
\begin{figure}[hbtp]
\centering
{\includegraphics[height=3.0cm]{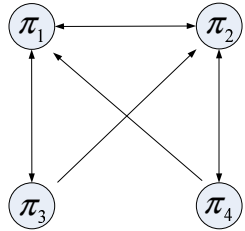}}
\caption{The corresponding edge-interconnection digraph of Figure \ref{figure:exam}.}
\label{figure:exam1}
\end{figure}
\begin{remark}
We use $\mu \left( t \right)$ as the only implementable control input for each agent, while use $u_e$ in the following analysis as a matter of convenience. However, we cannot presume the existence of $u \left( t \right)$ because \eqref{viutrual:controllaw} may not have a solution. Actually, whether or not the equation has a solution depends on the topological structure of the graph, and details will be discussed in the next section.
\end{remark}

By translating into the edge Laplacian dynamics, the novel system establishes a new control interconnection relation. To describe this relationship, we provide the following definition.

\begin{definition}[\emph{edge-interconnection digraph}]
Considering the ${\tilde x}_{e_k }$-subsystems as nodes and the control interconnections as directed edges, the interconnected system composed of the ${\tilde x}_{e_k } $-subsystems can be modeled as a digraph $\tilde {\mathcal{G}}$ called as \emph{edge-interconnection digraph}.
\end{definition}

In fact, $\tilde {\mathcal{G}}$ can be easily constructed from  $\mathcal{G}$, and equation (\ref{edge:finaldynamics}) illustrates the connection of the edge-interconnection digraph. In particular, there are two steps to transform $\mathcal{G}$ into $\tilde {\mathcal{G}}$. First, all the edges $e_k$ in $\mathcal{G}$ are modeled as the vertices and denoted by $\pi_{k}$. Second, base on (\ref{edge:finaldynamics}) and the definition of $A_e$, for any specified edge pair $e_k, e_l$, if $e_k, e_l \in  \mathcal{E}^i_\mathcal{C}$, then $\pi_{k}$ and $\pi_{l}$ are connected by a bidirectional edge; and if $e_k$ = $\mathcal{P}^i$ and $e_l = \mathcal{C}^i$, there will be a directed edge incident into $\pi_{k}$ from $\pi_{l}$.  Figure \ref{figure:exam1} depicts the corresponding edge-interconnection digraph of the example shown in Figure \ref{figure:exam}. To reveal the intricate relation between the original digraph and the edge-interconnection digraph, we need the following Lemma.



\begin{lemma}\label{Lemma:treetoGe}
For a rooted tree, the corresponding edge-interconnection digraph $\tilde  {\mathcal{G}}$ is composed of several strongly connected subgraphs $\tilde  {\mathcal{G}_i}, i = 1,2,\cdots$, which are coupled through simple cascaded connections and parallel connections.
\end{lemma}

\begin{proof}
Obviously, the sibling-edges $\mathcal{E}^i_\mathcal{C}$ with the same parent $i$ form a strongly connected components in the edge-interconnection digraph, since they affect each other mutually. Besides, for the $e_k$ = $\mathcal{P}^i$, and $e_l = \mathcal{C}^i$, there will be a directed edge incident into $\pi_{k}$ from $\pi_{l}$ \ul{which} forms the cascaded connection in the same branch. When taking the strongly connected components as nodes, $\tilde  {\mathcal{G}}$ is acyclic and consists of several cascaded connections and parallel connections. For instance, consider a rooted tree in Figure \ref{fig:rootedtree},  the corresponding edge-interconnection digraph is illustrated in Figure \ref{fig:inducedrootedtree}.
\end{proof}

\begin{figure}[hbtp]
\begin{center}
\mbox{\subfigure[A simple rooted tree.]
{\includegraphics[height=40mm]{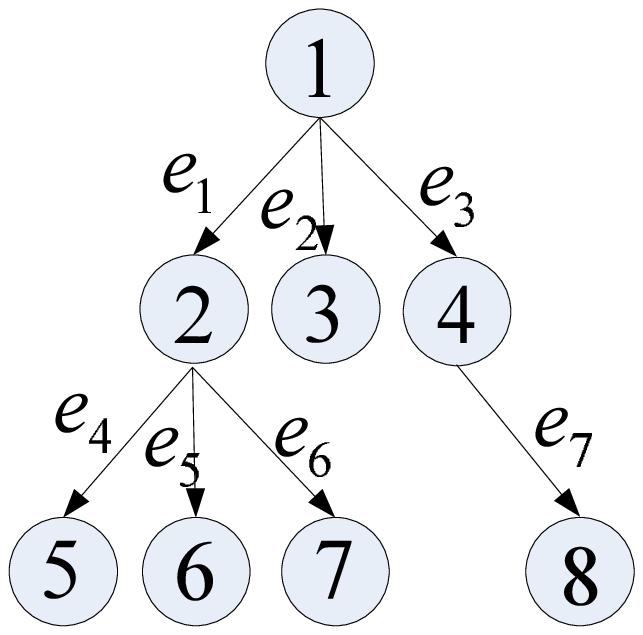}\label{fig:rootedtree}} \hspace{1cm}
\subfigure[The edge-interconnection digraph.]
{\includegraphics[height=40mm]{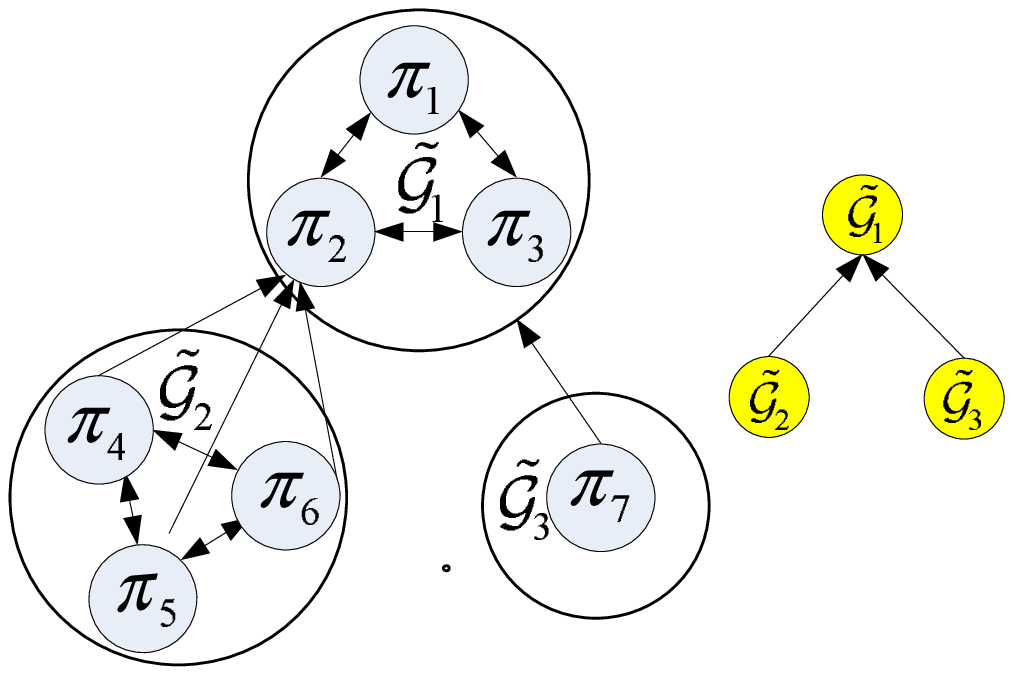}\label{fig:inducedrootedtree}}\quad
}
\caption{The corresponding edge-interconnection digraph with 3 strongly connected components $\tilde {\mathcal{G}}_{1},~\tilde {\mathcal{G}}_{2}$ and $\tilde {\mathcal{G}}_{3}$.}
\end{center}
\end{figure}

\begin{lemma}\label{lemma:stronglyconnectedISS}
\hl{The edge-interconnection graph $\tilde{ \mathcal{G}}$ is strongly connected if and only if  $\mathcal{G}$ is strongly connected digraph.}
\end{lemma}

\begin{proof}
Clearly, for the strongly connected graph, there will be a directed path connecting any pair of edges $(e_k, e_l)$ in each direction. In that way, any two distinct nodes $(\pi_k, \pi_l)$ of $\tilde {\mathcal{G}}$ can be connected via a directed path; therefore, $\tilde {\mathcal{G}}$  is strongly connected. \hl{On the other hand, while $\tilde {\mathcal{G}}$ is strongly connected, it implies that, for any pair of nodes of $\mathcal{G}$, there always exists a directed path connecting them, i.e., $\mathcal{G}$ is also strongly connected.}
\end{proof}

\subsection{Main Results}

In this section, both strongly connected digraphs and quasi-strongly connected digraphs are considered. The cyclic-small-gain theorem is then employed to guarantee the robust consensus of the closed-loop multi-agent systems. Note that ISS cyclic-small-gain theorem can be directly applied if the underlying digraph is strongly connected from \cite{liu2011lyapunov}. However, for the quasi-strongly case, we will translate it into a two-subsystem interconnection structure.

\subsubsection{Strongly Connected Digraph}\label{case:sc}

A digraph $\mathcal{G}$ of a six-agent system is shown as an example in Figure \ref{fig:simulation-stronglyconnected}, and the corresponding edge-interconnection digraph is shown in Figure \ref{fig:edge-simulation-stronglyconnected}. From Lemma \ref{lemma:stronglyconnectedISS}, we know that the edge-interconnection digraph is strongly connected if $\mathcal{G}$ is strongly connected.

\begin{figure}[hbtp]
\begin{center}
\mbox{\subfigure[The strongly connected  digraph.]
{\includegraphics[height=38mm]{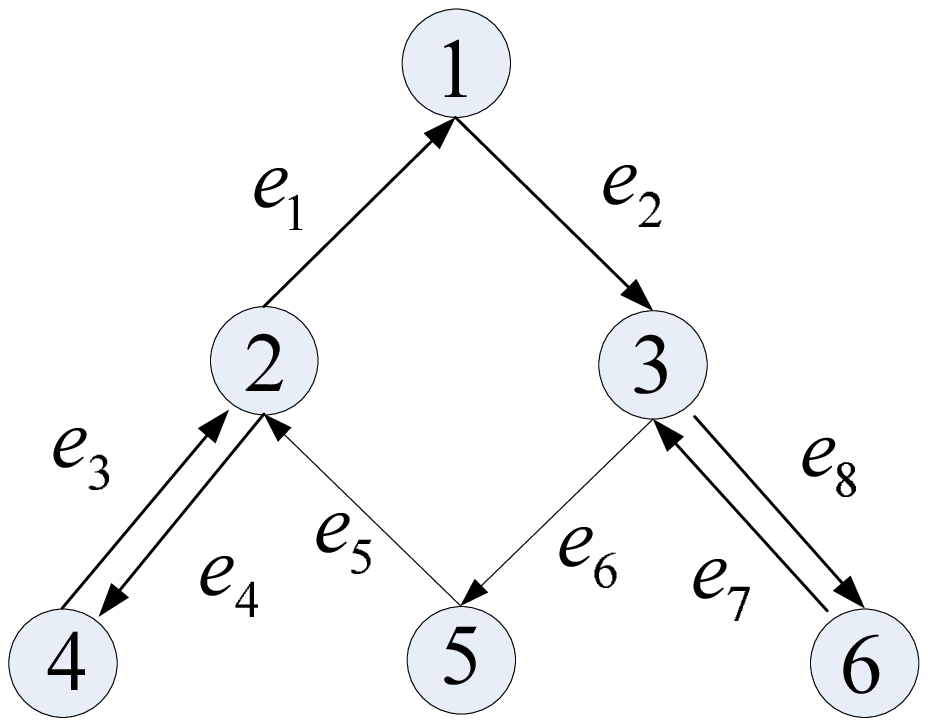}\label{fig:simulation-stronglyconnected}} \hspace{1cm}
\subfigure[The edge-interconnection digraph.]
{\includegraphics[height=45mm]{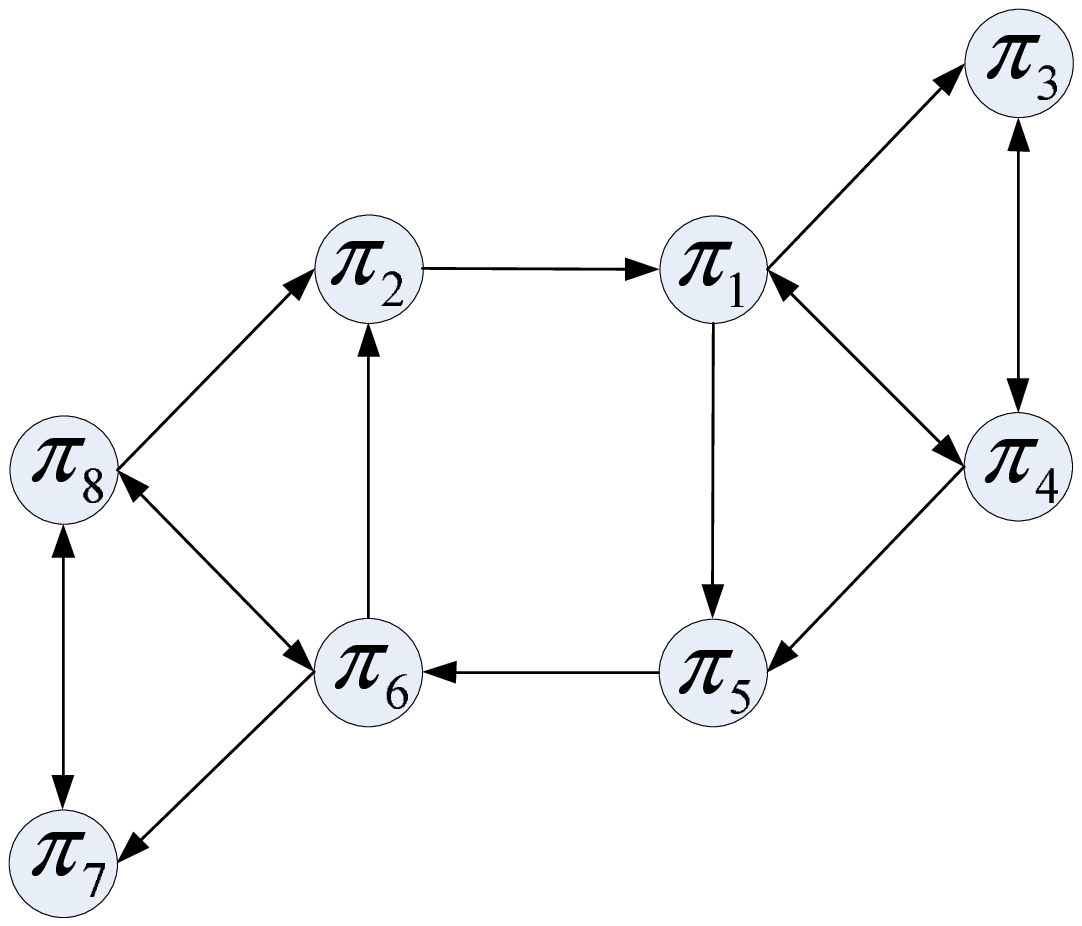}\label{fig:edge-simulation-stronglyconnected}}\quad
}
\caption{The original strongly connected digraph and the corresponding edge-interconnection digraph.}
\end{center}
\end{figure}

To begin our analysis, we define the following ISS-Lyapunov function candidate
\begin{align*}
V_{e_k }= {1 \over {2}} \tilde  x_{e_k }^T \tilde  x_{e_k },~~k =
1,2,\cdots,L
\end{align*}
and we denote $\Theta$ as the set of all simple loops (more details please refers to \cite{liu2011lyapunov}) of $\mathcal{G}$, and $A_o(\gamma _{\tilde x_{e_k } }^{\tilde x_{e_l } })$ as the product of the gain assigned to the edges of a simple loop $o \in \Theta$, where $\gamma _{\tilde x_{e_k } }^{\tilde x_{e_l } } \in K_\infty$  with $k = 1,2,\cdots, L$ and ${e_l  \in N_{e_k }^ \otimes }$.

The main result for the strongly connected digraph is given as follows.

\begin{theorem}\label{theorem:sc}
Assuming that the digraph is strongly connected, consider the subsystem (\ref{edge:finaldynamics}) with ${\tilde x_{e_k} }$ as the state, and $\tilde x_{e_l }\left( {e_l  \in N_{e_k }^ \otimes } \right)$ as the external inputs. For any specified constant $\sigma _{e_k }$ and $\gamma _{\tilde x_{e_k } }^{\tilde x_{e_l } }  \in K_\infty$  (${e_l  \in N_{e_k }^ \otimes }$), we can design
\begin{align}\label{virtual:controllawdesign}
u_{e_k }= - {{\tilde x_{e_k } } \over {\left| {\tilde x_{e_k } } \right|}}\left[ {\eta \left| {\tilde x_{e_k } } \right| + \sum\limits_{e_l  \in N_{e_k }^ \otimes } {\left| {\left[ {A_e } \right]_{kl} } \right|} \rho _{\tilde x_{e_l } }^{\tilde x_{e_k } } \left( {\left| {\tilde x_{e_k } } \right|} \right)} + \sqrt{2}\xi \right] + \left( {1 - {{\sigma _{e_k } } \over 2}} \right)\tilde x_{e_k }  \nonumber \\
\end{align}
with $\rho _{\tilde x_{e_l } }^{\tilde x_{e_k } }  = \underline{\alpha} ^{ - 1}  \circ \left( {\gamma _{\tilde x_{e_k } }^{\tilde x_{e_l } } } \right)^{ - 1}  \circ \overline{\alpha} \left( s \right)$ such that the subsystem (\ref{edge:finaldynamics}) is  ISS with ISS-Lyapunov function $V_{e_k }$ satisfying
$$
V_{e_k }  \ge \mathop {\max }\limits_{e_l  \in N_{e_k }^ \otimes } \left\{ {\gamma _{\tilde x_{e_k } }^{\tilde x_{e_l } } \left( {V_{e_l } } \right)} \right\} \Rightarrow \nabla V_{e_k } \dot {\tilde x}_{e_k }  \le  - \sigma _{e_k } V_{e_k }.
$$

If the required cyclic-small-gain condition (\ref{align:cyclic-small-gain}) is satisfied, the composed system is ISS.
Then, the objective robust consensus can be achieved by using the distributed consensus protocol (\ref{protocolA1}) with
\begin{equation}\label{auxiliary:controllaw}
{\rm{u = }}[{ E}^T]^{\dag} u_e.
\end{equation}
\end{theorem}
\begin{proof}

Since $V_{e_k }  \ge \mathop {\max }\limits_{e_l  \in N_{e_k }^\otimes } \left\{ {\gamma _{\tilde x_{e_k } }^{\tilde x_{e_l } } \left( {V_{e_l } } \right)} \right\}$, we have
\begin{equation*}
\left| {\tilde x_{e_l } } \right| \le \rho _{\tilde x_{e_l } }^{\tilde x_{e_k } } \left( {\left| {\tilde x_{e_k } } \right|} \right),{e_l  \in N_{e_k }^ \otimes. }
\end{equation*}

Note that $\left|\left[E^T\right]_k\right| =  \sqrt{2} $ and by taking the derivative of $V_{e_k }$, we have
\begin{align*}
\nabla V_{e_k } \dot {\tilde x}_{e_k }
& =\tilde x_{e_k }^T \dot {\tilde x}_{e_k } \\
& = \tilde x_{e_k}^T[ u_{e_k }+ \left( {f \left( {t,x_{\otimes \left( e_k \right) } } \right) - f \left( {t,x_{\odot \left( e_k \right) } } \right)} \right)+ [E^T]_k w   - \left( {\tilde x_{e_k }  + \sum\limits_{e_l  \in N_{e_k }^ \otimes } {\left[ {A_e } \right]_{kl} } \tilde x_{e_l } } \right)]  \\
& \le \tilde x_{e_k }^T \left( {u_{e_k } - \tilde x_{e_k } } \right) + \eta\left| {\tilde x_{e_k }^T } \right|\left| {\tilde x_{e_k } } \right| + \sqrt{2}\xi \left| {\tilde x_{e_k }^T } \right| - \tilde x_{e_k }^T \sum\limits_{e_l  \in N_{e_k }^ \otimes } {\left[ {A_e } \right]_{kl} } \tilde x_{e_l } \\
& \le \tilde x_{e_k }^T \left( {u_{e_k } - \tilde x_{e_k } } \right) + \eta\left| {\tilde x_{e_k }^T } \right|\left| {\tilde x_{e_k } } \right| +\sqrt{2}\xi \left| {\tilde x_{e_k }^T } \right|  + \left| {\tilde x_{e_k }^T } \right|\sum\limits_{e_l  \in N_{e_k }^ \otimes } {\left| {\left[ {A_e } \right]_{kl} } \right|} \rho _{\tilde x_{e_l } }^{\tilde x_{e_k } } \left( {\left| {\tilde x_{e_k } } \right|} \right).
\end{align*}
Using (\ref{virtual:controllawdesign}), we have
\begin{align*}
\nabla V_{e_k } \dot {\tilde x}_{e_k }  \le  - {{\sigma _{e_k } } \over 2} \tilde x_{e_k }^T \tilde x_{e_k }  =  - \sigma _{e_k } V_{e_k }
\end{align*}
which implies that the $\tilde {x}_k$-subsystem is ISS.

Since the induced edge-interconnection digraph is strongly connected, the ISS cyclic-small-gain theorem can be directly implemented. If the following cyclic-small-gain condition is satisfied
\begin{align*}
 A_o(\gamma _{\tilde x_{e_k } }^{\tilde x_{e_l } })< \mathrm{Id}
\end{align*}
then the composed system \eqref{edge:subsystem} is ISS.

It should be mentioned that as $u_e$ is designed, the whole system ${\tilde x_e\left( t \right)}$ is unforced. Based on the ISS property, we have
\begin{equation*}
\left| {\tilde x_e\left( t \right)} \right| \le \tilde \beta \left( {\left| {\tilde x^0 } \right|,t} \right) + \tilde \gamma \left( {\left\| w \right\|_\infty  } \right),t \ge 0
\end{equation*}
where $\tilde x_e^0$ is the initial state of $\tilde x_e\left( t \right)$ and $\tilde \beta \in \kl $, $\tilde \gamma \in \kf$. Obviously, as $t \rightarrow \infty$, we have $\tilde \beta \left( {\left| {\tilde x^0 } \right|,t} \right) = 0$. So that
\begin{align*}
\mathop {\lim }\limits_{t \to \infty } \left| {\tilde x_e\left( t \right)} \right| \le \tilde \gamma \left( {\left\| w \right\|_\infty  } \right),
\end{align*}
which implies the robust consensus. Since $\mathcal{G}$ has a spanning tree, the pseudoinverse $[{ E}^T]^{\dag}$ exists. Then we can obtain the implemented consensus control input by using \eqref{auxiliary:controllaw}. The proof is concluded.
\end{proof}

\subsubsection{Quasi-strongly Connected Digraph}\label{case:qsc}

\begin{figure}[hbtp]
\begin{center}
\mbox{\subfigure[The quasi-strongly connected  digraph.]
{\includegraphics[height=38mm]{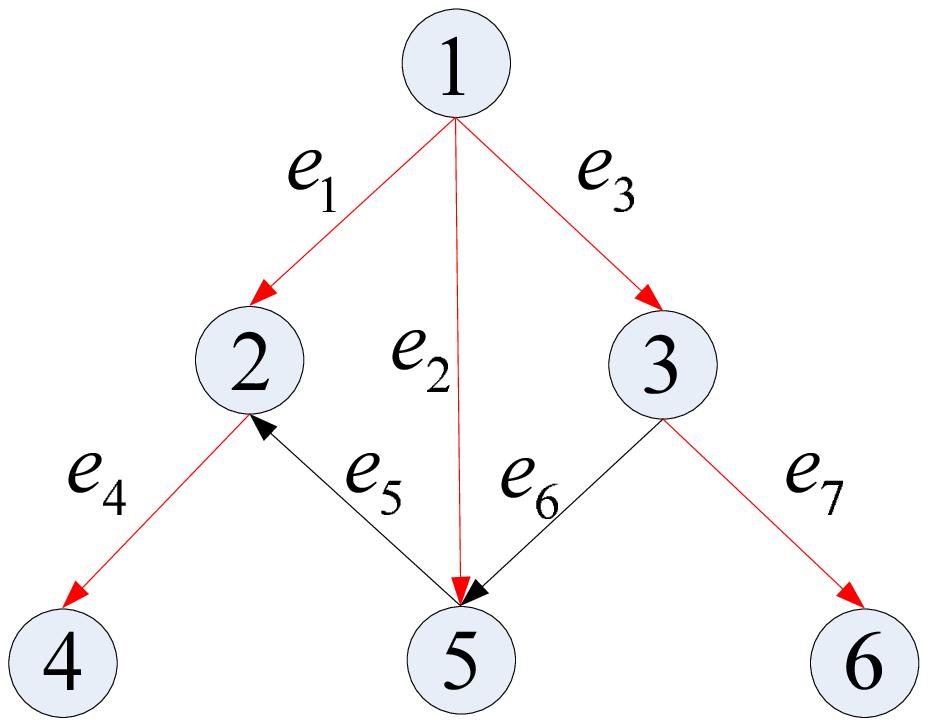}\label{fig:simulation-spanningtree}} \hspace{1cm}
\subfigure[The edge-interconnection digraph.]
{\includegraphics[height=42mm]{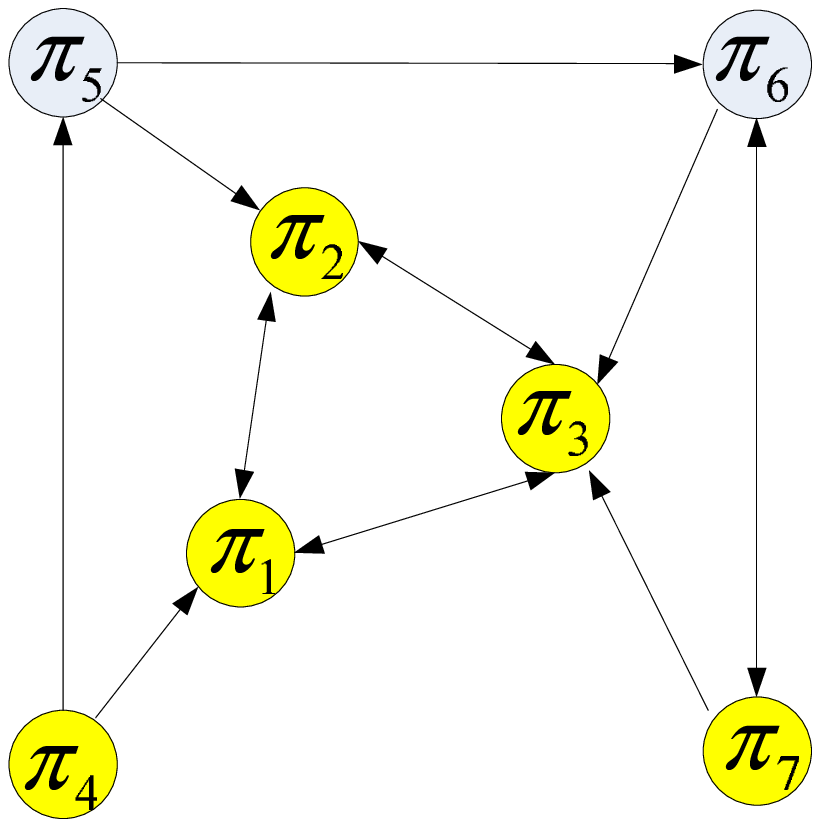}\label{fig:edge-simulation-spanningtree}}\quad
}
\caption{The quasi-strongly connected digraph and the corresponding edge-interconnection digraph.}
\end{center}
\end{figure}

In this section, we consider a quasi-strongly connected digraph $\mathcal{G} = \mathcal{G}_\mathcal{T}  \cup \mathcal{G}_\mathcal{C}$. An example is given as Figure \ref{fig:simulation-spanningtree}. The edges of $\mathcal{G}_\mathcal{T}$ are marked \ul{as} red. Accordingly, the edge-interconnection digraph is shown in Figure \ref{fig:edge-simulation-spanningtree}, which  consists of two parts: $\tilde {\mathcal{G}} = \tilde{\mathcal{G}}_\mathcal{T}  \cup \tilde{\mathcal{G}}_\mathcal{C}$. Correspondingly, the edge Laplacian dynamics system $\tilde x_e(t) $ can be modeled as the interconnection of the $H_\mathcal{T}$-subsystem and the $H_\mathcal{C}$-subsystem based on $\tilde{\mathcal{G}}_\mathcal{T}$ and $\tilde{\mathcal{G}}_\mathcal{C}$.

As previously mentioned, the incidence matrix can be rewritten as $ E = \left[ {\begin{matrix}{E_\mathcal{T} } & {E_\mathcal{C} } \end{matrix}} \right]$, and the edge Laplacian can be represented as the block form $ L_e = \left[ {\begin{matrix}
   {{ L_{e1}}} & {{ L_{e2}}}  \cr
   {{ L_{e3}}} & {{ L_{e4}}}  \cr
\end{matrix} } \right]$ in line with the permutation. Therefore, the edge Laplacian dynamics $\tilde x_e(t)$ described by (\ref{edge:subsystem}) can be translated into the following form:
\begin{align*}
H_\mathcal{T} :\dot {\tilde x}_\mathcal{T} \left( t \right) = \mathcal{F}_\mathcal{T}\left( {t,x} \right) -  L_{e_1 } \tilde x_\mathcal{T} \left( t \right) -  L_{e_2 } \tilde x_\mathcal{C} \left( t \right) +  E_\mathcal{T}^T w \left( t \right) +  u_\mathcal{T} \left( t \right)
\end{align*}
with $\mathcal{F}_\mathcal{T}\left( {t,x} \right) =   E_\mathcal{T}^T \mathcal{F}\left( {t,x} \right)$, $u_\mathcal{T}\left( t \right) = E_\mathcal{T}^T u\left( t \right) $.
\begin{align*}
H_\mathcal{C}: \dot {\tilde x}_\mathcal{C} \left( t \right) = \mathcal{F}_\mathcal{C}\left( {t,x} \right) -  L_{e_4 } \tilde x_\mathcal{C} \left( t \right) -  L_{e_3 } \tilde x_\mathcal{T} \left( t \right)+  E_\mathcal{C}^T w \left( t \right) + u_\mathcal{C}\left( t \right)
\end{align*}
with $\mathcal{F}_\mathcal{C}\left( {t,x} \right) = E_\mathcal{T}^T \mathcal{F}\left( {t,x} \right)$, $u_\mathcal{C}\left( t \right) = E_\mathcal{C}^T u\left( t \right) $. Besides, $E_\mathcal{T}^T w$ and $E_\mathcal{C}^T w$ indicate unknown but bounded disturbances on $\tilde{\mathcal{G}}_\mathcal{T}$ and $\tilde{\mathcal{G}}_\mathcal{C}$, respectively. The interacting system is shown in Figure \ref{figure:spanning-cospanning interconnection}.

\begin{figure}[hbtp]
\centering
{\includegraphics[height=3.0cm]{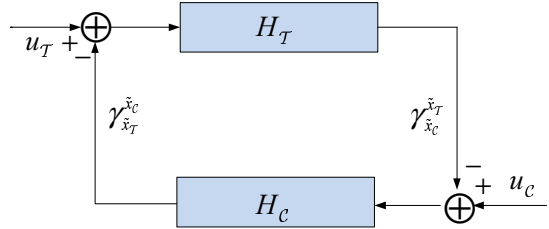}}
\caption{The interconnection of subsystems $H_t$ and $H_c$.}
\label{figure:spanning-cospanning interconnection}
\end{figure}

Obviously, for a specific spanning tree, it contains $N-1$ edges. For subsystem $\tilde x_{\mathcal{T}_k}, \text{~with~} k =1,2,\cdots,N-1$, $e_l ,e_k  \in \mathcal{G}_\mathcal{T}$, we have
\begin{align}\label{align:spanningtreext}
\dot {\tilde x}_{\mathcal{T}_k}   = F_{\mathcal{T}_k } - \left( {\tilde x_{\mathcal{T}_k }  + \sum\limits_{ e_l  \in N_{e_k }^ \otimes }  {\left[ { {A}_{e_1} } \right]_{kl} } \tilde x_{\mathcal{T}_l }} \right)  - Le_{2_k } \tilde x_\mathcal{C} +  \left[ E_\mathcal{T}^T\right]_k w  + u_{\mathcal{T}_k }.
\end{align}

We choose the following ISS-Lyapunov candidates:
\begin{align}
& V_{\mathcal{T}_k}  = {1 \over 2} \tilde{x}_{\mathcal{T}_k}^T \tilde {x}_{\mathcal{T}_k},~~~ k = 1,2,\cdots,N-1 \label{lyap:Vtk} \\
& V_\mathcal{T}  = {1 \over 2} \tilde x_\mathcal{T}^T \tilde x_\mathcal{T}\label{lyap:Vt} \\
& V_\mathcal{C}  = {1 \over 2} \tilde x_\mathcal{C}^T \tilde x_\mathcal{C}. \label{lyap:Vc}
\end{align}

Denote $\Theta_\mathcal{T}$ as the set of all the simple loops of $\tilde{\mathcal{G}}_\mathcal{T}$, and denote $A_{\mathcal{T}_o}(\gamma _{\tilde x_{t_k } }^{\tilde x_{t_l } })$ as the product of the gain assigned to the edges of a simple loop in $\Theta_\mathcal{T}$. Let $\bar \lambda_1$ denotes the smallest nonzero eigenvalue of matrix $TT^T$, where $T$ is defined in \eqref{euqation:T}.

We then present the main result for the quasi-strongly connected digraph as follows:
\begin{theorem}\label{lemma:subsystemsISSforspanningtree}
Assuming that the digraph is quasi-strongly connected, consider the subsystem (\ref{align:spanningtreext}) with ${\tilde x_{\mathcal{T}_k } }\left(e_k  \in \mathcal{G}_\mathcal{T} \right)$ as the internal state. Let $\tilde x_{\mathcal{T}_l}\left( {e_l\in {N_{e_k }^ \otimes \cap \mathcal{G}}_\mathcal{T}} \right) $ and $ \tilde x_\mathcal{C} \left( e_c  \in  {\mathcal{G}}_\mathcal{C} \right)$ be the external inputs. For any specified constant $\sigma _{ \mathcal{T}_k } > 0$ and $\gamma _{\tilde x_{\mathcal{T}_k } }^{ \tilde x_{\mathcal{T}_l } }, \gamma _{\tilde x_{\mathcal{T}_k } }^{\tilde x_\mathcal{C} }  \in k_\infty$, we can design
\begin{align}\label{virtual:controllawdesignforspanningtree}
&  u_{\mathcal{T}_k }= \nonumber \\
&  - {{\tilde x_{\mathcal{T}_k } } \over {\left| {\tilde x_{\mathcal{T}_k } } \right|}}\left[ {\eta \left| {\tilde  x_{\mathcal{T}_k } } \right| + \sum\limits_{e_l  \in N_{e_k }^ \otimes } {\left| {\left[ { A_{e_1} } \right]_{kl} } \right|} \rho _{\tilde  x_{\mathcal{T}_l } }^{\tilde  x_{\mathcal{T}_k } } \left( {\left| {\tilde  x_{\mathcal{T}_k } } \right|} \right)}+  \left|  L_{{e_2}_k}  \right| \rho _{\tilde  x_{\mathcal{C} } }^{\tilde  x_{\mathcal{T}_k }} \left( {\left| {\tilde  x_{\mathcal{T}_k } } \right|} \right) + \sqrt{2}\xi \right] + \left( {1 - {{\sigma _{\mathcal{T}_k } } \over 2}} \right) \tilde  x_{\mathcal{T}_k }
\end{align}
with $\rho _{\tilde x_{\mathcal{T}_l } }^{\tilde x_{\mathcal{T}_k } }  = \underline{\alpha} ^{ - 1}  \circ \left( {\gamma _{ \tilde x_{\mathcal{T}_k } }^{\tilde x_{\mathcal{T}_l } } } \right)^{ - 1}  \circ \overline{\alpha} \left( s \right)$ and $\rho _{\tilde x_\mathcal{C} }^{\tilde x_{\mathcal{T}_k } }  = \underline{\alpha} ^{ - 1}  \circ \left( {\gamma _{ \tilde x_{\mathcal{T}_k } }^{\tilde x_\mathcal{C} } } \right)^{ - 1}  \circ \overline{\alpha} \left( s \right)$, such that the subsystem (\ref{align:spanningtreext}) is  ISS with an ISS-Lyapunov function $V_{\mathcal{T}_k }$ satisfying
\begin{align*}
V_{\mathcal{T}_k }  \ge \mathop {\max }\limits_{ e_l  \in N_{e_k }^ \otimes} \left\{ {\gamma _{\tilde x_{\mathcal{T}_k } }^{\tilde x_{\mathcal{T}_l } } \left( {V_{\mathcal{T}_l } } \right) }, {\gamma _{\tilde x_{\mathcal{T}_k } }^{\tilde x_\mathcal{C} } \left( {V_\mathcal{C} } \right) } \right\} \Rightarrow \nabla V_{\mathcal{T}_k } \dot {\tilde x}_{\mathcal{T}_k }  \le  - \sigma _{\mathcal{T}_k } V_{\mathcal{T}_k }.
\end{align*}
If the cyclic-small-gain condition \eqref{align:cyclic-small-gain} as well as
\begin{align}\label{align:smallgain}
\gamma _{\tilde x_{\mathcal{T}_k } }^{\tilde x_\mathcal{C} }  < {1 \over {{ \bar \lambda_1} \left( {N - 1} \right)}}
\end{align}
are satisfied, the whole system is ISS. Then the objective robust consensus can be achieved by using the protocol \eqref{protocolA1} with
\begin{equation}\label{auxiliary:controllaw1}
{\rm{u = }}\left[ { E_\mathcal{T}^T} \right]^{\dag}  u_\mathcal{T}.
\end{equation}
\end{theorem}

\begin{proof}
The main proof procedure \ul{contains} three steps. Firstly, for the ${\tilde x}_{\mathcal{T}_k}$-subsystem, the ISS properties can be guaranteed by taking \eqref{virtual:controllawdesignforspanningtree} as the control law. Second, the ISS properties of the upside subsystem $H_\mathcal{T}$ are then proven by utilizing the ISS cyclic-small-gain theorem. Finally, we prove that the downside subsystem $H_\mathcal{C}$ is ISS, and the objective robust consensus can be achieved while the small gain condition $\gamma _{\tilde x_\mathcal{T} }^{\tilde x_\mathcal{C} }  \circ \gamma _{\tilde x_\mathcal{C} }^{\tilde x_\mathcal{T} }  < \mathrm{Id}$ is satisfied.

\emph{Step 1}:

Using the Lyapunov candidate defined in (\ref{lyap:Vtk}) and considering $V_{\mathcal{T}_k } \ge \mathop {\max }\limits_{ e_l  \in N_{e_k }^ \otimes } \left\{ {\gamma _{\tilde x_{\mathcal{T}_k } }^{\tilde x_{\mathcal{T}_l } } \left( {V_{\mathcal{T}_l } } \right) }, {\gamma _{\tilde x_{\mathcal{T}_k } }^{\tilde x_{\mathcal{C}} } \left( {V_{\mathcal{C} } } \right) } \right\}$ with ${e_l ,e_k  \in \mathcal{G}_\mathcal{T}, e_c  \in \mathcal{G}_\mathcal{C} }$, we have
\begin{align*}
& \left| {\tilde x_\mathcal{C} } \right| \le \rho _{\tilde x_\mathcal{C} }^{\tilde x_{\mathcal{T}_k } } \left( {\left| {\tilde x_{\mathcal{T}_k } } \right|} \right)  \\
& \left| {\tilde x_{\mathcal{T}_l } } \right| \le \rho _{\tilde x_{\mathcal{T}_l } }^{\tilde x_{\mathcal{T}_k } } \left( {\left| {\tilde x_{\mathcal{T}_k } } \right|} \right).
\end{align*}

Taking the derivative of $V_{\mathcal{T}_k }$,  we have
\begin{align*}
\nabla V_{\mathcal{T}_k } \dot {\tilde x}_{\mathcal{T}_k }
& = \tilde x_{\mathcal{T}_k }^T \dot {\tilde x}_{\mathcal{T}_k } \\
& = \tilde x_{\mathcal{T}_k}^T\left[ u_{\mathcal{T}_k }+ {f\left( {t,x_{\otimes\left( e_k \right)}} \right) - f \left( {t,x_{\odot\left( e_k \right)} } \right)}  - \left( {\tilde x_{\mathcal{T}_k }  + \sum\limits_{ e_l  \in N_{e_k }^ \otimes} {\left[ { A_{e_1} } \right]_{kl} } \tilde x_{\mathcal{T}_l } } \right) -  Le_{2_k } \tilde x_\mathcal{C} +\left[  E_\mathcal{T}^T\right]_k w  \right] \\
& \le \tilde x_{\mathcal{T}_k }^T \left( {u_{\mathcal{T}_k }  - \tilde x_{\mathcal{T}_k } } \right) + \eta \left| {\tilde x_{\mathcal{T}_k }^T } \right|\left| {\tilde x_{\mathcal{T}_k } } \right| - \tilde x_{\mathcal{T}_k }^T \sum\limits_{ e_l  \in N_{e_k }^ \otimes } {\left[ { A_{e_1} } \right]_{kl} } \tilde x_{\mathcal{T}_l } - \tilde x_{\mathcal{T}_k }^T  Le_{2_k } \tilde x_\mathcal{C} + \tilde x_{\mathcal{T}_k }^T\left[ E_\mathcal{T}^T\right]_k w \\
& \le \tilde x_{\mathcal{T}_k }^T \left( {u_{\mathcal{T}_k }  - \tilde x_{\mathcal{T}_k } } \right) + \eta \left| {\tilde x_{\mathcal{T}_k }^T } \right|\left| {\tilde x_{\mathcal{T}_k } } \right| + \left| {\tilde x_{\mathcal{T}_k }^T } \right|\sum\limits_{e_l  \in N_{e_k }^ \otimes }{\left| {\left[ { A_{e_1} } \right]_{kl} } \right|} \rho _{\tilde x_{\mathcal{T}_l } }^{\tilde x_{\mathcal{T}_k } } \left( {\left| {\tilde x_{\mathcal{T}_k } } \right|} \right) +\left| {\tilde x_{\mathcal{T}_k }^T } \right| \left|  L_{{e_2}_k}  \right| \rho _{\tilde  x_\mathcal{C} }^{\tilde  x_{\mathcal{T}_k }} \left( {\left| {\tilde  x_{\mathcal{T}_k } } \right|} \right) \\
& \quad + \sqrt{2}\xi\left| {\tilde x_{\mathcal{T}_k }^T } \right| .
\end{align*}
By using (\ref{virtual:controllawdesignforspanningtree}), we obtain
\begin{align*}
\nabla V_{\mathcal{T}_k } \dot {\tilde x}_{\mathcal{T}_k }  \le  - {{\sigma _{\mathcal{T}_k } } \over 2} \tilde x_{\mathcal{T}_k }^T \tilde x_{\mathcal{T}_k }  =  - \sigma _{\mathcal{T}_k } V_{\mathcal{T}_k }
\end{align*}
which implies that the ${\tilde x}_{\mathcal{T}_k}$-subsystem is ISS.

\emph{Step 2}:
We define $\tilde x_{s_i} = \{\tilde x_{\mathcal{T}_k}:\text{$e_k  \in \tilde{\mathcal{G}}_\mathcal{T} \text{~and~} e_k \in
\mathcal{E}^i_\mathcal{C}$} \}$ as the state of the strongly connected component. By taking  \eqref{virtual:controllawdesignforspanningtree} as the input, each subsystem $\tilde x_{\mathcal{T}_k}$ is ISS and admits an ISS-Lyapunov function $V_{\mathcal{T}_k}$. From Lemma \ref{lem:CSG}, for the set of all the simple loops $\Theta_\mathcal{T}$, if the cyclic-small-gain condition
\begin{align*}
A_{\mathcal{T}_o}(\gamma _{\tilde x_{\mathcal{T}_k } }^{\tilde x_{\mathcal{T}_l } })  < \mathrm{Id}
\end{align*}
is satisfied, then the composed subsystem $\tilde x_{s_i}$ is ISS. By taking the strongly connected subsystems $\tilde x_{s_i}$ as nodes, the upside subsystem $H_\mathcal{T}$ is acyclic as we previously mentioned in Lemma \ref{Lemma:treetoGe}. From \cite{tanner2004leader}, we note that the ISS properties are retained if the underlying digraph is acyclic. Therefore, the upside subsystem $H_\mathcal{T}$ is ISS as well.

Additionally, we can verify that $V_\mathcal{T}$ defined in \eqref{lyap:Vt} is an ISS-Lyapunov function. Also we can calculate the the interconnection gain $\gamma _{\tilde x_\mathcal{T} }^{\tilde x_\mathcal{C} }$ from $H_\mathcal{C}$ to $H_\mathcal{T}$. To begin with, according to
\begin{align*}
V_{\mathcal{T}_k } \ge \mathop {\max }\limits_{ e_l  \in N_{e_k }^ \otimes } \left\{ {\gamma _{\tilde x_{\mathcal{T}_k } }^{\tilde x_{\mathcal{T}_l } } \left( {V_{\mathcal{T}_l } } \right) }, {\gamma _{\tilde x_{\mathcal{T}_k } }^{\tilde x_{\mathcal{C}} } \left( {V_{\mathcal{C} } } \right) } \right\}
\end{align*}
we have
\begin{align*}
V_\mathcal{T} \ge (N-1) {\gamma _{\tilde x_{\mathcal{T}_k } }^{\tilde x_{\mathcal{C}} } \left( {V_{\mathcal{C} } } \right) }, k = 1,2, \cdots, N - 1
\end{align*}
since $V_\mathcal{T}  = \sum\limits_{k = 1}^{N -1} {V_{\mathcal{T}_k } }$.

By choosing $\sigma_k = \sigma > 0$, we obtain
\begin{align*}
\nabla V_{\mathcal{T}} \dot {\tilde x}_{\mathcal{T}} & =  \sum\limits_{k = 1}^{N -1} \nabla V_{\mathcal{T}_k} \dot {\tilde x}_{\mathcal{T}_k}  \le  - \sum\limits_{k = 1}^{N -1} {{\sigma_k } \over 2}  \tilde x_{\mathcal{T}_k }^T \tilde x_{\mathcal{T}_k }  =  - \sigma  V_{\mathcal{T}}
\end{align*}
which implies $V_{\mathcal{T}_k }$ is an ISS-Lyapunov function. Then we can simply choose $\gamma _{\tilde x_\mathcal{T} }^{\tilde x_\mathcal{C} }$ as
\begin{align}\label{gain1}
\gamma _{\tilde x_\mathcal{T} }^{\tilde x_\mathcal{C} } = (N-1) {\gamma _{\tilde x_{\mathcal{T}_k } }^{\tilde x_\mathcal{C} }}.
\end{align}

\emph{Step 3}:
Since we have $\tilde x_\mathcal{C}  = T^T \tilde x_\mathcal{T}$ from (\ref{equation:xc-xt}), then
\begin{align*}
V_\mathcal{C}  & = {1 \over 2}\tilde x_\mathcal{C}^T \tilde x_\mathcal{C} = {1 \over 2}\tilde x_\mathcal{T}^T TT^T \tilde x_\mathcal{T}
\end{align*}
where $TT^T$ is symmetric positive semidefinite. Suppose the eigenvalues of $TT^T$ can be ordered and denoted as $0 \le \bar \lambda_1 \le \bar \lambda_2\le\cdots\le \bar \lambda_{N-1}$. Clearly, one can obtain that $V_{\mathcal{C}}  \ge \bar \lambda_1 V_{\mathcal{T}}$.

Assume that $P$ is an orthogonal transformation matrix and let $\tilde x_\mathcal{T} = Py$, then we can translate $V_\mathcal{C}$ into a standard quadratic form as follows:
\begin{align*}
V_{\mathcal{C}} ={1 \over 2}( \bar \lambda_1 y_1^2 + \bar \lambda_2 y_2^2 + \cdots + \bar \lambda_{N-1} y_{N-1}^2).
\end{align*}
By taking the derivation of $V_\mathcal{C}$, one can obtain
\begin{align*}
\nabla V_{\mathcal{C}} \dot {\tilde x}_{\mathcal{C}} & =  \sum\limits_{k = 1}^{N -1} \bar \lambda_k y_k^T \dot y_k
=  \sum\limits_{k = 1}^{N -1}\bar \lambda_k {\tilde x_{\mathcal{T}}}^T P_k P_k^{-1} {\dot {\tilde x}_{\mathcal{T}}} \nonumber \\
& =  \sum\limits_{k = 1}^{N -1}\bar \lambda_k \nabla V_{\mathcal{T}} \dot {\tilde x}_{\mathcal{T}}  \le  - \sigma  \sum\limits_{k = 1}^{N -1}\bar \lambda_k V_{\mathcal{T}} \le  - {\sigma \over \bar \lambda_{N-1}} \sum\limits_{k = 1}^{N -1} \bar \lambda_k V_{\mathcal{C}}
\end{align*}
which implies that $V_\mathcal{C}$ is an ISS-Lyapunov function. Then we can choose the interconnection gain as
\begin{align}\label{gain2}
\gamma _{\tilde x_\mathcal{C} }^{\tilde x_\mathcal{T} } = \bar \lambda_1.
\end{align}
For this two interacting subsystems $H_\mathcal{T}$ and $H_\mathcal{C}$, if the small gain condition $\gamma _{\tilde x_\mathcal{T} }^{\tilde x_\mathcal{C} }  \circ \gamma _{\tilde x_\mathcal{C} }^{\tilde x_\mathcal{T} }  < \mathrm{Id}$ is hold, then the whole system is ISS. To satisfy the small gain condition, by combining \eqref{gain1} and \eqref{gain2}, we can choose
\begin{align}\label{cond1}
\gamma _{\tilde x_{\mathcal{T}_k } }^{\tilde x_\mathcal{C} }  < {1 \over {{\bar \lambda_1} \left( {N - 1} \right)}}.
\end{align}
From Theorem \ref{theorem:sc}, it is clear that the objective robust consensus can be guaranteed while \eqref{cond1} is hold.



\end{proof}
\begin{remark}
The ISS-Lyapunov function for the composite system $\tilde x\left( t \right)$ can be obtained by using the approach mentioned in \cite{liu2011lyapunov}. Besides, the discussion about the explicit cyclic-small-gain conditions required in (\ref{align:cyclic-small-gain}) can be found in our previous study \cite{wang2012second}. In particular, we can check that the cyclic-small-gain conditions can be guaranteed by simply choosing the nonlinear gains $\gamma _{\tilde x_{e_k } }^{\tilde x_{e_l } } < \mathrm{Id}$, with ${e_l  \in N_{e_k }^ \otimes }$.
\end{remark}


\section{SIMULATIONS}
Numerical simulations are performed to illustrate the obtained theoretical results. For this set of simulations, we consider a six-agent system with both strongly connected graph and quasi-strongly connected graph. The dynamics of the $i$-th agent is assumed to be
\begin{equation*}
\dot x_i\left(t\right)   = f \left( {t,x_i } \right)  + w_i\left(t\right) + u_i\left(t\right),~~~i = 1,2, \cdots,6
\end{equation*}
where $x_i\left(t\right), u_i\left(t\right) \in \rtn^3$ , with the inherent nonlinear dynamics $f :\rtn \times \rtn^3 \to \rtn^3$ described by Chua's circuit
\begin{align}\label{align:chaos}
f \left(t, {x_i } \right) =
(\zeta \left( { - x_{i1}  + x_{i2}  - l\left( {x_{i1} } \right)} \right),x_{i1}  - x_{i2}  + x_{i3} , - \chi x_{i2} )^T
\end{align}
where $l\left( {x_{i1} } \right) = bx_{i1}  + 0.5\left( {a - b} \right)\left( {\left| {x_{i1}  + 1} \right| - \left| {x_{i1}  - 1} \right|} \right)$. While choosing $\zeta  = 10$, $\chi  = 18$, $a =  - 4/3$ and $b =  - 3/4$, system (\ref{align:chaos}) is chaotic with the Lipschitz constant $\eta  = 4.3871$ (\cite{yu2010second}) as shown in Figure \ref{figure:chaos}. Assume the state of each agent is corrupted by white noise $w_i\left(t\right) \in \rtn^3 $ with the noise power $\xi  = \left[ {\begin{matrix} {0.25} & {0.25} & {0.25} \end{matrix} } \right]^T$.

\begin{figure}[hbtp]
\centering
{\includegraphics[height=4.5cm]{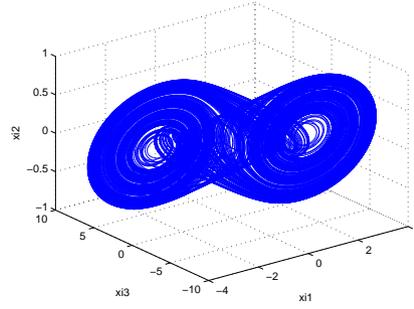}}
\caption{Chaos with two attractors.}
\label{figure:chaos}
\end{figure}

\subsection{Case 1: Strongly Connected}

The digraph is strongly connected as shown in Figure \ref{fig:simulation-stronglyconnected}. From (\ref{edge:finaldynamics}),  multi-agent system can be translated into the edge Laplacian dynamics associated with the edge-interconnection digraph shown in Figure \ref{fig:edge-simulation-stronglyconnected}. The \ul{incidence} matrix $E$ and  the edge \ul{adjacency} matrix $A_e$ are
$$
E = \left( {\begin{matrix}
   { - 1} & 1 & 0 & 0 & 0 & 0 & 0 & 0  \cr
   1 & 0 & { - 1} & 1 & { - 1} & 0 & 0 & 0  \cr
   0 & { - 1} & 0 & 0 & 0 & 1 & { - 1} & 1  \cr
   0 & 0 & 1 & { - 1} & 0 & 0 & 0 & 0  \cr
   0 & 0 & 0 & 0 & 1 & { - 1} & 0 & 0  \cr
   0 & 0 & 0 & 0 & 0 & 0 & 1 & { - 1}  \cr
 \end{matrix} } \right)
$$
$$
Ae = \left( {\begin{matrix}
   0 & { - 1} & 0 & 1 & 0 & 0 & 0 & 0  \cr
   0 & 0 & 0 & 0 & 0 & { - 1} & 0 & { - 1}  \cr
   { - 1} & 0 & 0 & { - 1} & 0 & 0 & 0 & 0  \cr
   1 & 0 & { - 1} & 0 & 0 & 0 & 0 & 0  \cr
   { - 1} & 0 & 0 & { - 1} & 0 & 0 & 0 & 0  \cr
   0 & 0 & 0 & 0 & { - 1} & 0 & 0 & 1  \cr
   0 & 0 & 0 & 0 & 0 & { - 1} & 0 & { - 1}  \cr
   0 & 0 & 0 & 0 & 0 & 1 & { - 1} & 0  \cr
 \end{matrix} } \right).
$$
Then the edge Laplacian can be calculated though $Le = I + Ae$.

By simply choosing $\gamma_{\tilde x_{e_k} }^{\tilde x_{e_l} }  < \mathrm{Id}~\text{with}~ k,l=1,2,\cdots,8, e_l  \in N_{e_k }^ \otimes$, the cyclic-small-gain theorem condition is satisfied. By taking $\gamma_{\tilde x_{e_k} }^{\tilde x_{e_l} }\left(s\right) = 0.9487s$ and $\underline{\alpha} = \overline{\alpha} = 1$, then we obtain
$$
\rho _{\tilde x_{e_l} }^{\tilde x_{e_k} }  = \underline{\alpha} ^{ - 1}  \circ \left( {\gamma _{\tilde x_{e_k} }^{\tilde x_{e_l} } } \right)^{ - 1}  \circ \overline{\alpha} \left( s \right)= 1.0541s.
$$

After choosing $\sigma _{e_k }  = 1$, the input for the edge-interconnection system \eqref{edge:subsystem} is proposed as
\begin{align*}
& u_{e_1 }  =  - {{\tilde x_{e_1} } \over {\left| {\tilde x_{e_1 } } \right|}}\left( {\eta\left| {\tilde x_{e_1} } \right| + 2.1082\left| {\tilde x_{e_1} } \right|} + \sqrt{2} \xi  \right) + 0.5\tilde x_{e_1}, u_{e_2 }  =  - {{\tilde x_{e_2 } } \over {\left| {\tilde x_{e_2 } } \right|}}\left( {\eta\left| {\tilde x_{e_2} } \right| + 2.1082\left| {\tilde x_{e_2} } \right|} + \sqrt{2} \xi \right) + 0.5\tilde x_{e_2} \\
&u_{e_3 }  =  - {{\tilde x_{e_3 } } \over {\left| {\tilde x_{e_3 } } \right|}}\left( {\eta\left| {\tilde x_{e_3} } \right| + 2.1082\left| {\tilde x_{e_3} } \right|} + \sqrt{2} \xi \right) + 0.5\tilde x_{e_3},  u_{e_4 }  =  - {{\tilde x_{e_4 } } \over {\left| {\tilde x_{e_4 } } \right|}}\left( {\eta\left| {\tilde x_{e_4} } \right| + 2.1082\left| {\tilde x_{e_4} } \right|} + \sqrt{2} \xi \right) + 0.5\tilde x_{e_4} \\
&u_{e_5 }  =  - {{\tilde x_{e_5 } } \over {\left| {\tilde x_{e_5 } } \right|}}\left( {\eta\left| {\tilde x_{e_5} } \right| + 2.1082\left| {\tilde x_{e_5} } \right|} + \sqrt{2} \xi \right) + 0.5\tilde x_{e_5},  u_{e_6 }  =  - {{\tilde x_{e_6 } } \over {\left| {\tilde x_{e_6 } } \right|}}\left( {\eta\left| {\tilde x_{e_6} } \right| + 2.1082\left| {\tilde x_{e_6} } \right|} + \sqrt{2} \xi \right) + 0.5\tilde x_{e_6}  \\
&u_{e_7 }  =  - {{\tilde x_{e_7 } } \over {\left| {\tilde x_{e_7 } } \right|}}\left( {\eta\left| {\tilde x_{e_7} } \right| + 2.1082\left| {\tilde x_{e_7} } \right|} + \sqrt{2} \xi \right) + 0.5\tilde x_{e_7},  u_{e_8 }  =  - {{\tilde x_{e_8 } } \over {\left| {\tilde x_{e_8 } } \right|}}\left( {\eta\left| {\tilde x_{e_8} } \right| + 2.1082\left| {\tilde x_{e_8} } \right|} +\sqrt{2} \xi \right) + 0.5\tilde x_{e_8}.
\end{align*}
Finally, by using (\ref{auxiliary:controllaw}), the consensus protocol (\ref{protocolA1}) can be obtained.

The simulation results are shown in Figure \ref{fig:stronglyconnected}. The edge states $\tilde x_{e_k}$ converge to the neighbors of the origin by applying the consensus protocol shown in Figure \ref{fig:edgevalues_stronglyconnected}. From Figure \ref{fig:stronglyconnected_x1}, \ref{fig:stronglyconnected_x2} and \ref{fig:stronglyconnected_x3}, one can see that the robust consensus is indeed achieved. Therefore, the proposed consensus protocol can effectively address the challenges resulting from the inherently nonlinear dynamics and unknown but bounded disturbances.
\begin{figure}[hbtp]
\begin{center}
\mbox{
\subfigure[The edge agreement is achieved under control law \eqref{virtual:controllawdesign} and \eqref{auxiliary:controllaw}.]
{\includegraphics[height=50mm]{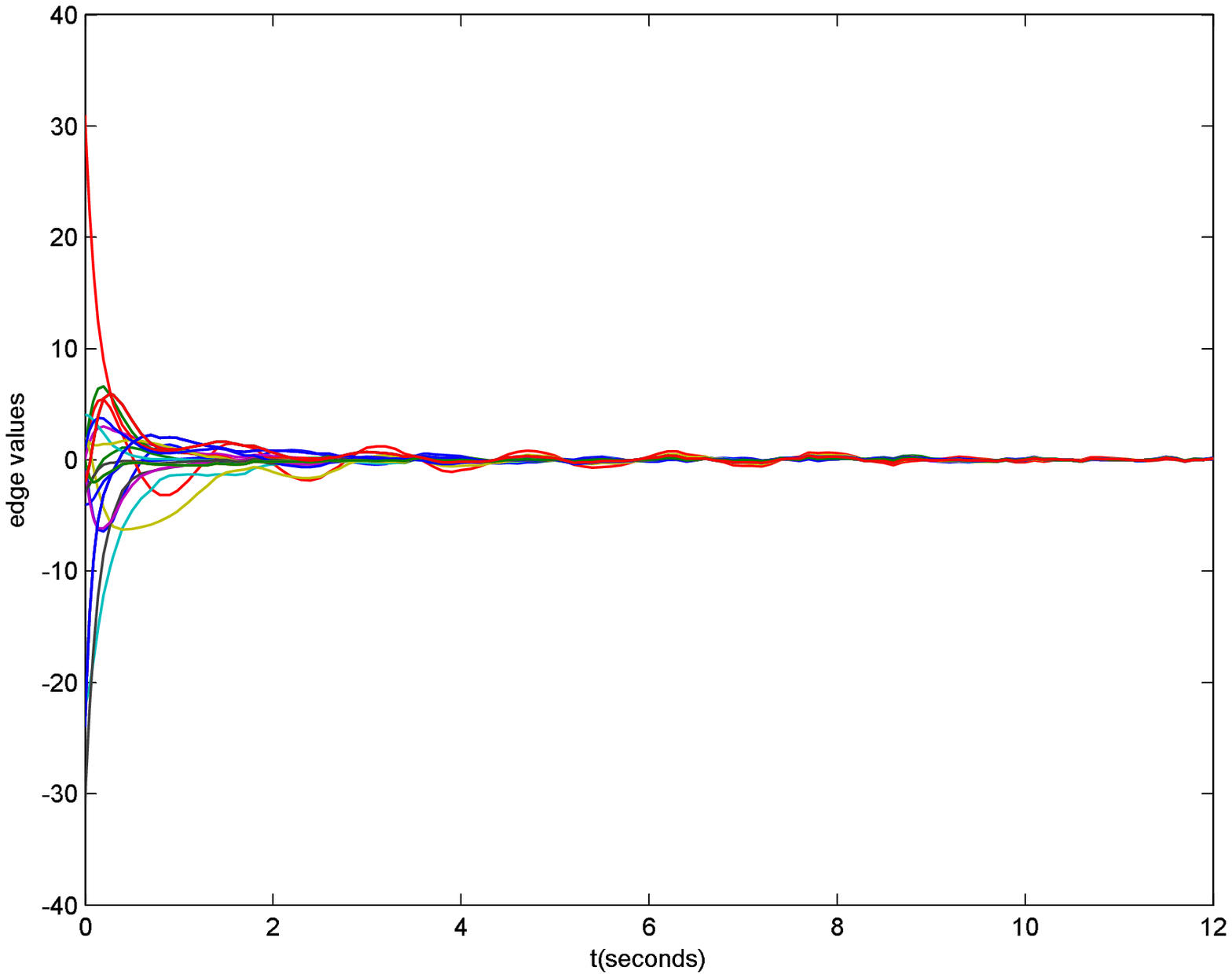}\label{fig:edgevalues_stronglyconnected}}\qquad
\subfigure[The evolutions of $x_{i1}$.]
{\includegraphics[height=50mm]{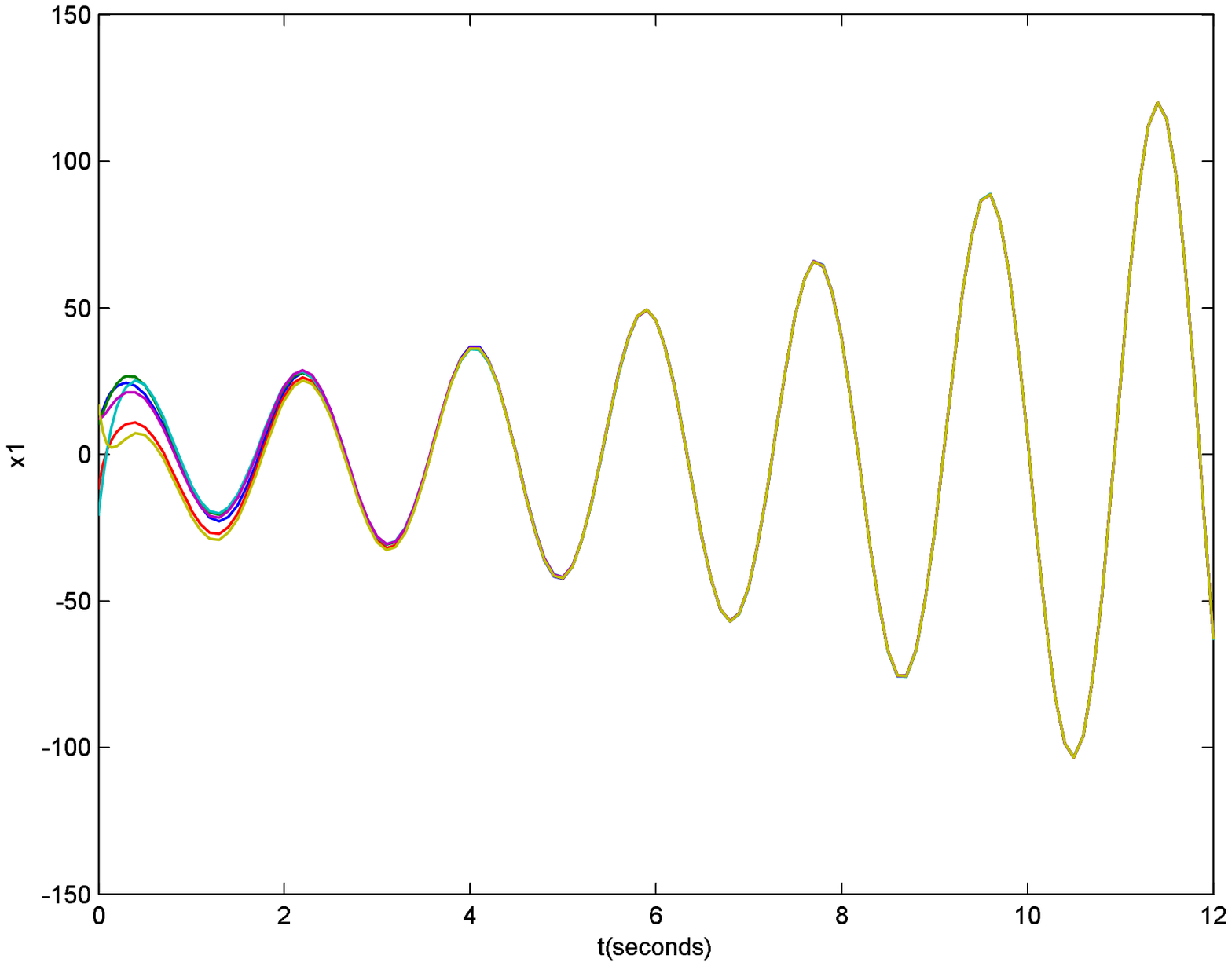}\label{fig:stronglyconnected_x1}}
}
\mbox{\subfigure[The evolutions of $x_{i2}$.]
{\includegraphics[height=50mm]{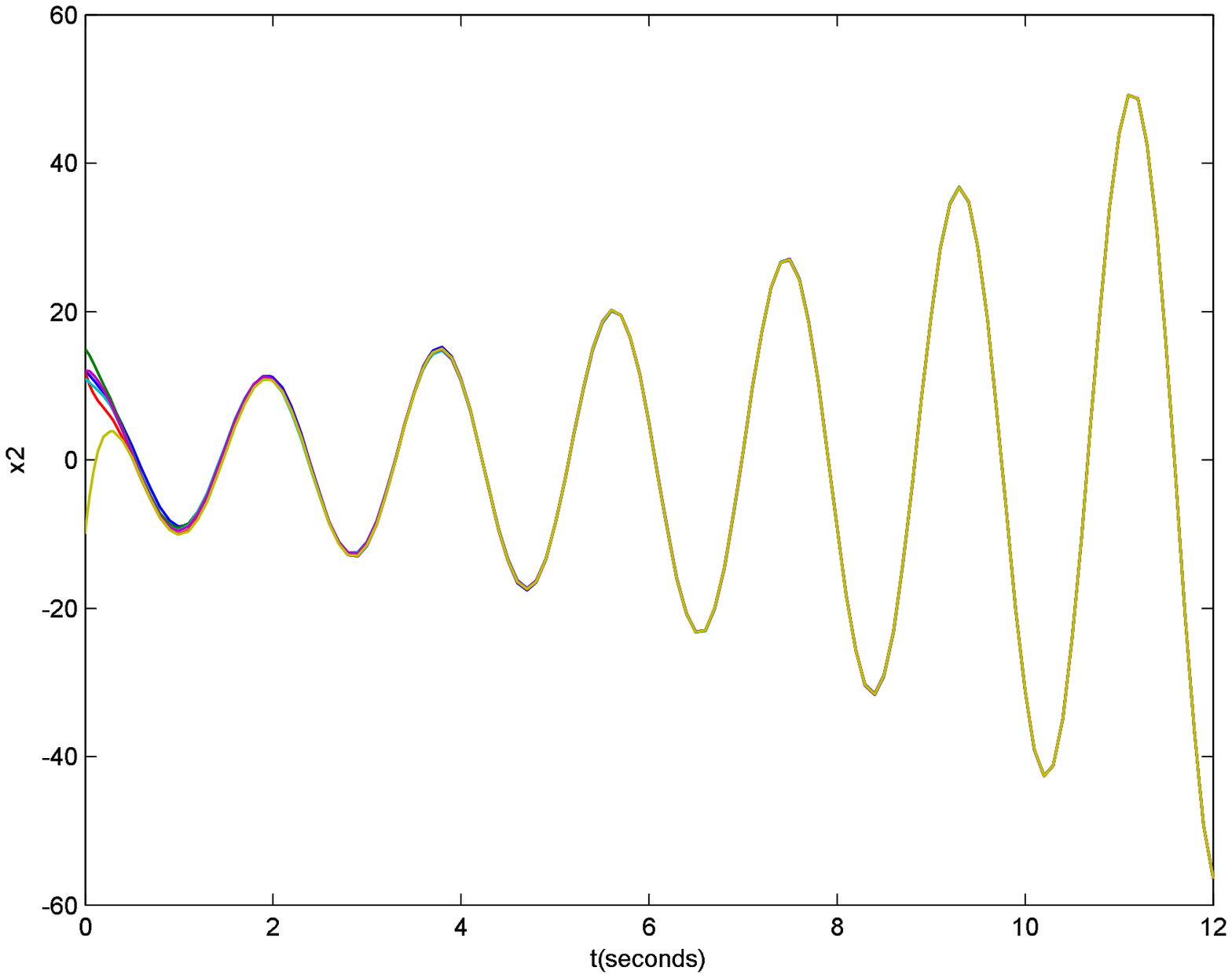}\label{fig:stronglyconnected_x2}}\qquad
\subfigure[The evolutions of $x_{i3}$.]
{\includegraphics[height=50mm]{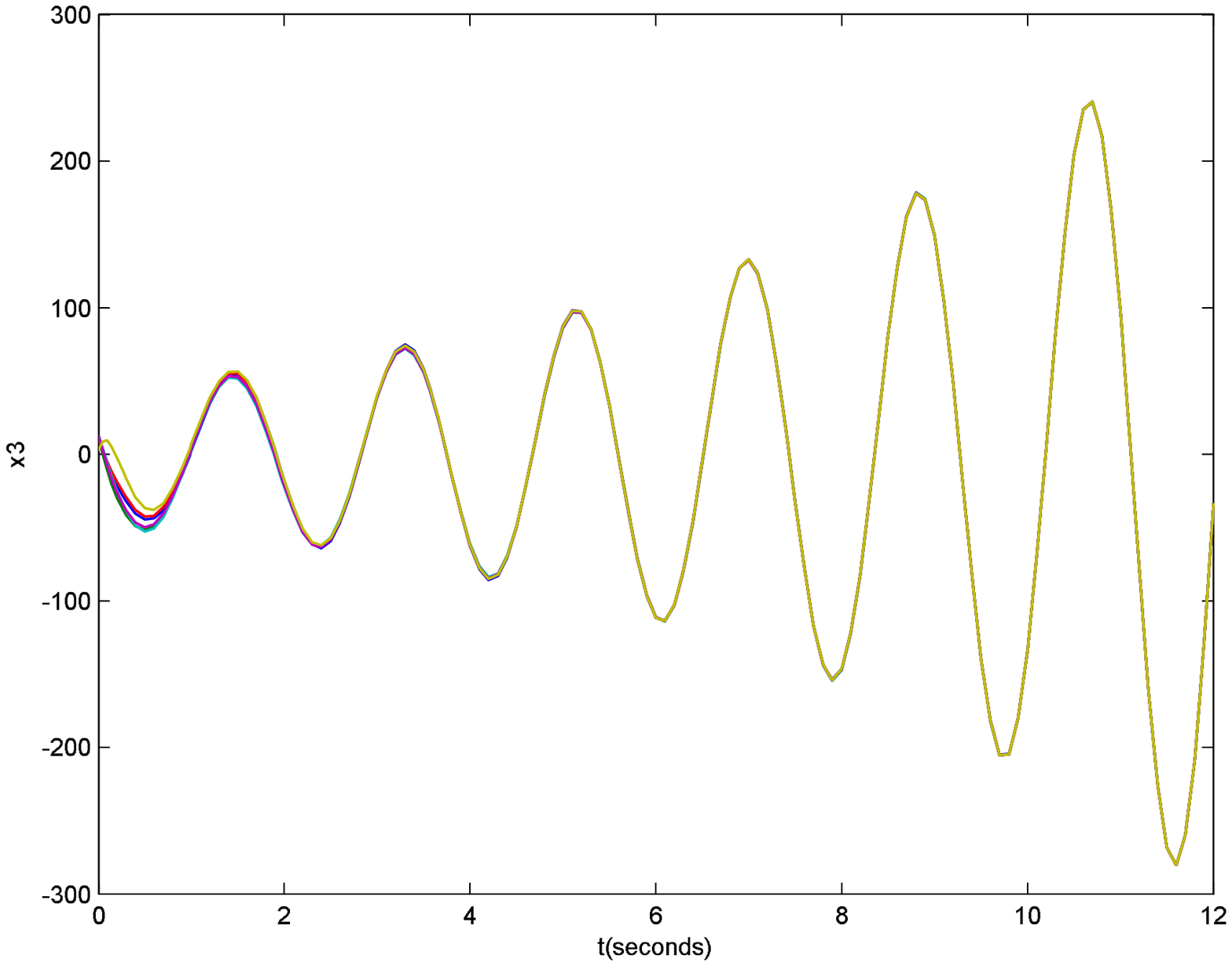}\label{fig:stronglyconnected_x3}}
}
\caption{The simulation results under strongly connected digraph.}
\label{fig:stronglyconnected}
\end{center}
\end{figure}

\subsection{Case 2: Quasi-strongly Connected}
In this case, Figure \ref{fig:simulation-spanningtree} depicts a quasi-strongly connected digraph, while Figure \ref{fig:edge-simulation-spanningtree} illustrates its corresponding edge-interconnection digraph. The yellow nodes in Figure \ref{fig:edge-simulation-spanningtree} correspond to $\tilde {\mathcal{G}}_\mathcal{T}$. According to the partition, the \ul{incidence} matrix of the spanning tree $ E_\mathcal{T}$, the \ul{incidence} matrix of the cospanning tree $ E_\mathcal{C}$  and the edge-adjacency matrix $ A_e$ are as follows:

$$
 E_\mathcal{T}  = \bordermatrix{
  {}&{e_1 } & {e_2 } & {e_3 } & {e_4 } & {e_7 }  \cr
  {}&1 & 1 & 1 & 0 & 0  \cr
  {}&{ - 1} & 0 & 0 & 1 & 0  \cr
  {}&0 & 0 & { - 1} & 0 & 1  \cr
  {}&0 & 0 & 0 & { - 1} & 0  \cr
  {}&0 & { - 1} & 0 & 0 & 0  \cr
  {}&0 & 0 & 0 & 0 & { - 1}  \cr
 }~~~~
 E_\mathcal{C}  = \bordermatrix{
   {}&{e_5 } & {e_6 }  \cr
   {}&0 & 0  \cr
   {}&{ - 1} & 0  \cr
   {}&0 & 1  \cr
   {}&0 & 0  \cr
   {}&1 & { - 1}  \cr
   {}&0 & 0  \cr
}
$$
$$
 A_e = \bordermatrix{
   {} & {e_1 } & {e_2 } & {e_3 } & {e_4 } & {e_7 } & {e_5 } & {e_6 }  \cr
   {e_1 } & 0 & 1 & 1 & { - 1} & 0 & 0 & 0  \cr
   {e_2 } & 1 & 0 & 1 & 0 & 0 & { - 1} & 0  \cr
   {e_3 } & 1 & 1 & 0 & 0 & { - 1} & 0 & { - 1}  \cr
   {e_4 } & 0 & 0 & 0 & 0 & 0 & 0 & 0  \cr
   {e_7 } & 0 & 0 & 0 & 0 & 0 & 0 & 1  \cr
   {e_5 } & 0 & 0 & 0 & { - 1} & 0 & 0 & 0  \cr
   {e_6 } & 0 & 0 & 0 & 0 & 1 & { - 1} & 0  \cr
}.
$$
By simply choosing $\gamma_{\tilde x_{\mathcal{T}_k} }^{\tilde x_{\mathcal{T}_l} }\left(s\right) = 0.9487s, k,l=1,2,3,4,7, e_l  \in N_{e_k }^ \otimes$, the cyclic-small-gain condition is satisfied. From (\ref{euqation:T}), we could have
$$
T = \left[ {\begin{matrix}
   1 & 0  \cr
   { - 1} & 1  \cr
   0 & { - 1}  \cr
   0 & 0  \cr
   0 & 0  \cr
 \end{matrix}} \right].
$$
The smallest non-zero eigenvalue of $TT^T$ is $\bar \lambda_1 = 1$; therefore, from (\ref{align:smallgain}), we can choose $\gamma _{\tilde x_{{\mathcal{T}_k}} }^{\tilde x_\mathcal{C} } = 0.175s$ and resulting
$$
\rho _{\tilde x_{\mathcal{T}_l} }^{\tilde x_{\mathcal{T}_k} }  = \underline{\alpha} ^{ - 1}  \circ \left( {\gamma _{\tilde x_{\mathcal{T}_k} }^{\tilde x_{\mathcal{T}_l} } } \right)^{ - 1}  \circ \overline{\alpha} \left( s \right)= 1.0541s
$$
$$
\rho _{\tilde x_\mathcal{C} }^{\tilde x_{{\mathcal{T}_k} } }  = \underline{\alpha} ^{ - 1}  \circ \left( {\gamma _{ \tilde x_{{\mathcal{T}_k} } }^{\tilde x_{\mathcal{C}} } } \right)^{ - 1}  \circ \overline{\alpha} \left( s \right) = 5.7143s.
$$
After taking $\sigma _{e_k }  = 1$, the control input for each edge-interconnection system \eqref{edge:subsystem} can be determined as
\begin{align*}
& u_{e_1 }  =  - {{\tilde x_{\mathcal{T}_1 } } \over {\left| {\tilde x_{\mathcal{T}_1 } } \right|}}\left( {\eta\left| {\tilde x_{\mathcal{T}_1} } \right| + 3.1623\left| {\tilde x_{\mathcal{T}_1} } \right|} + \sqrt{2} \xi \right) + 0.5\tilde x_{\mathcal{T}_1}\\
&u_{e_2 }  =  - {{\tilde x_{\mathcal{T}_2 } } \over {\left| {\tilde x_{\mathcal{T}_2 } } \right|}}\left( {\eta\left| {\tilde x_{\mathcal{T}_2} } \right| + 7.8225\left| {\tilde x_{\mathcal{T}_2} } \right|} + \sqrt{2} \xi \right) + 0.5\tilde x_{\mathcal{T}_2} \\
&u_{e_3 }  =  - {{\tilde x_{\mathcal{T}_3 } } \over {\left| {\tilde x_{\mathcal{T}_3 } } \right|}}\left( {\eta\left| {\tilde x_{\mathcal{T}_3} } \right| + 8.8766\left| {\tilde x_{\mathcal{T}_3} } \right|} + \sqrt{2} \xi \right) + 0.5\tilde x_{\mathcal{T}_3}\\
&u_{e_4 }  =  - {{\tilde x_{\mathcal{T}_4 } } \over {\left| {\tilde x_{\mathcal{T}_4 } } \right|}}\left( \eta\left| {\tilde x_{\mathcal{T}_4} } \right| + \sqrt{2} \xi \right) + 0.5\tilde x_{\mathcal{T}_4}  \\
&u_{e_7 }  =  - {{\tilde x_{\mathcal{T}_7 } } \over {\left| {\tilde x_{\mathcal{T}_7 } } \right|}}\left( {\eta\left| {\tilde x_{\mathcal{T}_7} } \right| + 5.7143\left| {\tilde x_{\mathcal{T}_7} } \right|}+\sqrt{2} \xi \right) + 0.5\tilde x_{\mathcal{T}_7}.
\end{align*}
Finally, based on (\ref{protocolA1}) and (\ref{auxiliary:controllaw1}), the implementable consensus protocol can be obtained.

Figure \ref{fig:edge_value_spanningtree} shows that the edge states reach agreement by using the consensus protocol. Simultaneously, multi-agent system achieves robust consensus shown in Figure \ref{fig:quasistrongly_x1}, \ref{fig:quasistrongly_x2} and \ref{fig:quasistrongly_x3}. Clearly, the proposed consensus protocol can effectively restrain the influences resulting from the inherently nonlinear dynamics and the unknown but bounded disturbances.
\begin{figure}[hbtp]
\begin{center}
\mbox{
\subfigure[The edge agreement is achieved under control law \eqref{virtual:controllawdesignforspanningtree} and \eqref{auxiliary:controllaw1}.]
{\includegraphics[height=50mm]{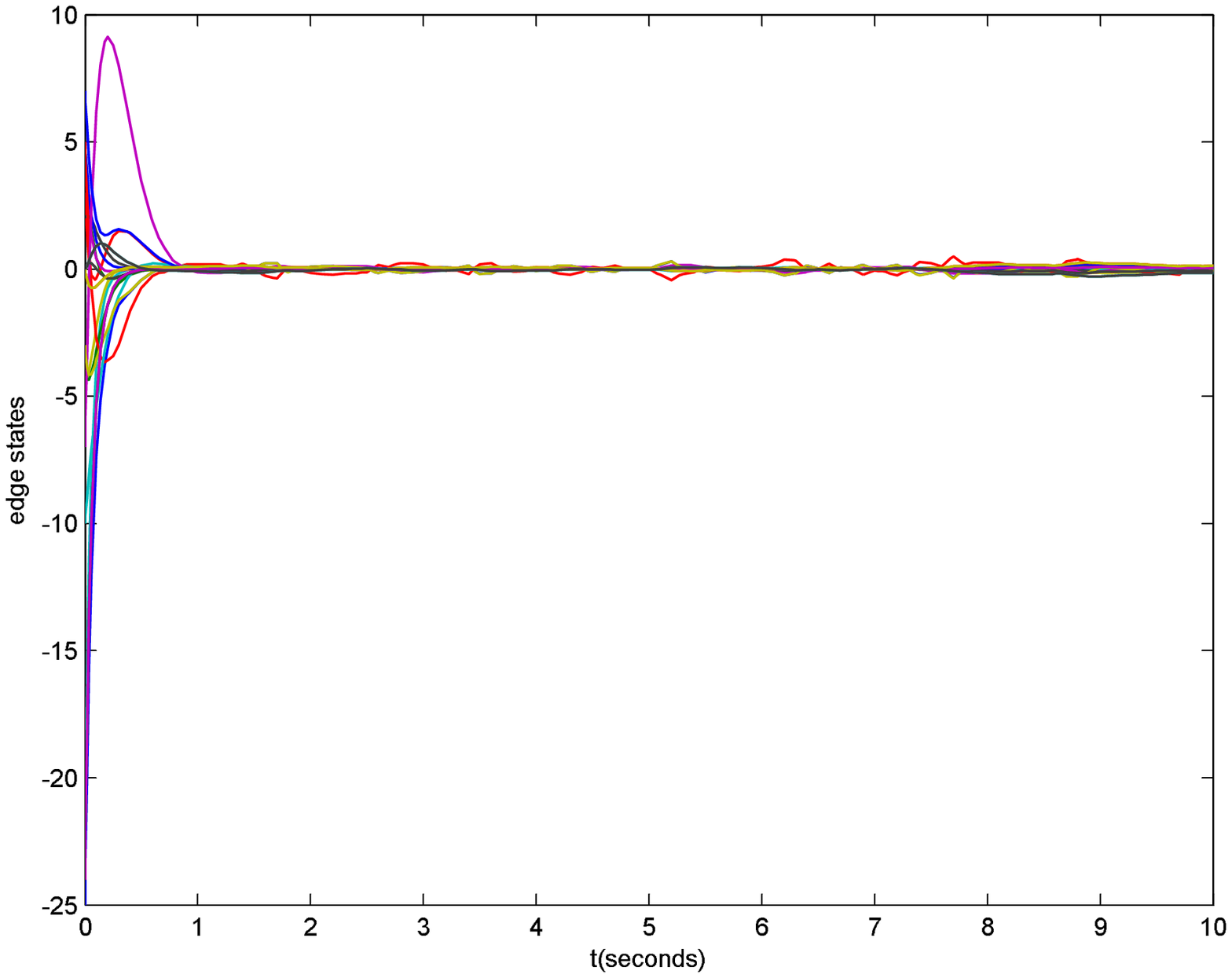}\label{fig:edge_value_spanningtree}} \qquad
\subfigure[The evolutions of $x_{i1}$.]
{\includegraphics[height=50mm]{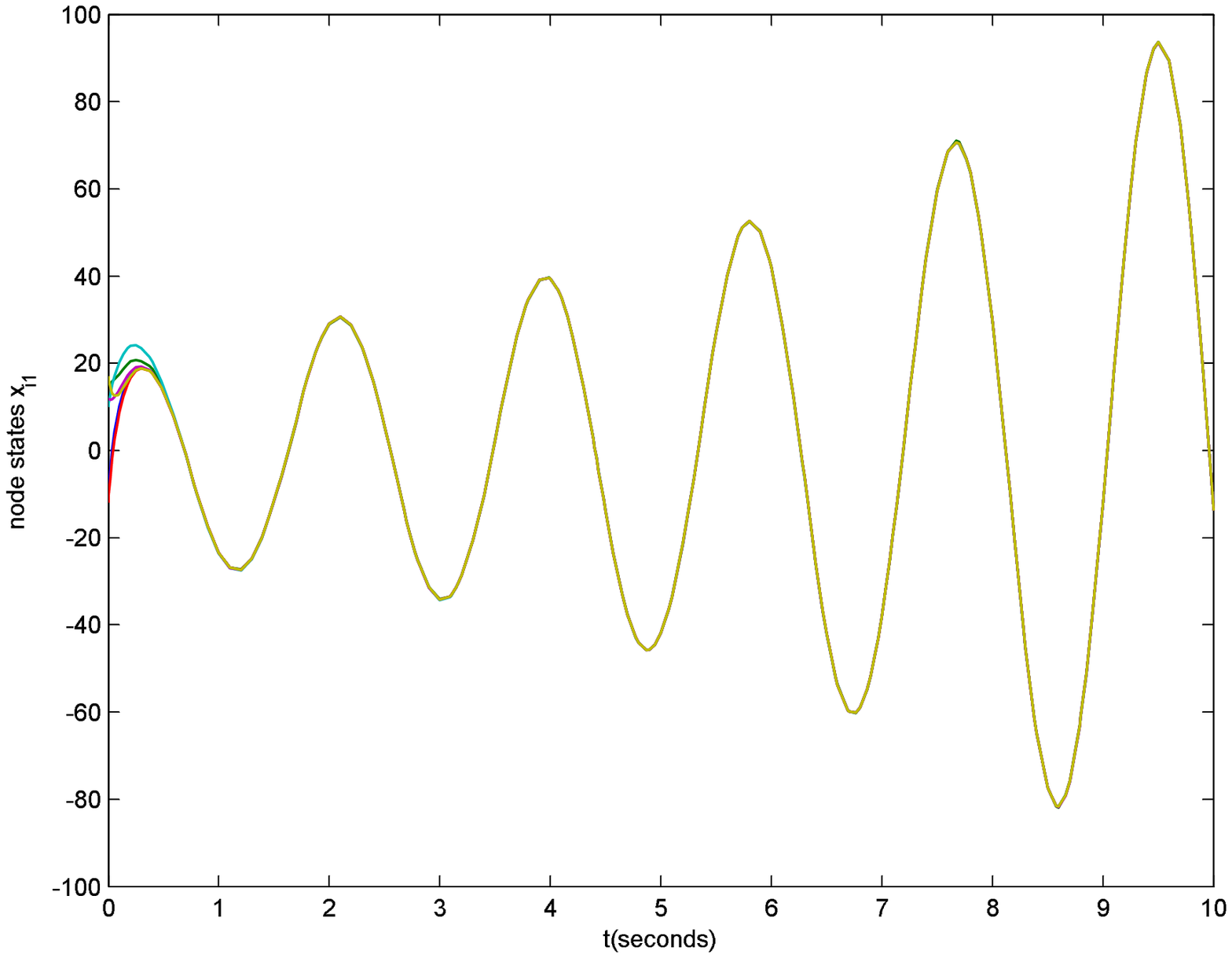}\label{fig:quasistrongly_x1}}
}
\mbox{\subfigure[The evolutions of $x_{i2}$.]
{\includegraphics[height=50mm]{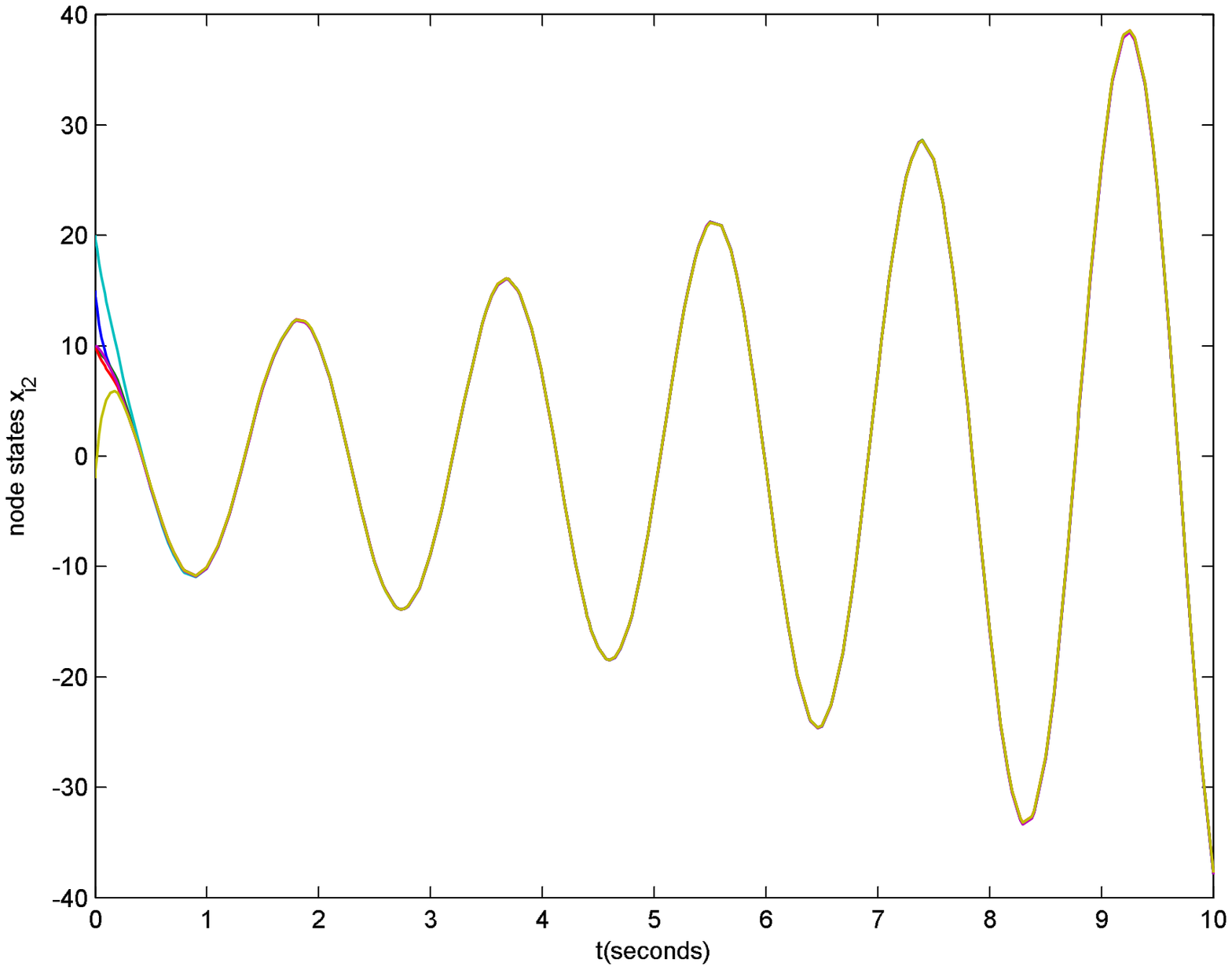} \label{fig:quasistrongly_x2}}\qquad
\subfigure[The evolutions of $x_{i3}$.]
{\includegraphics[height=50mm]{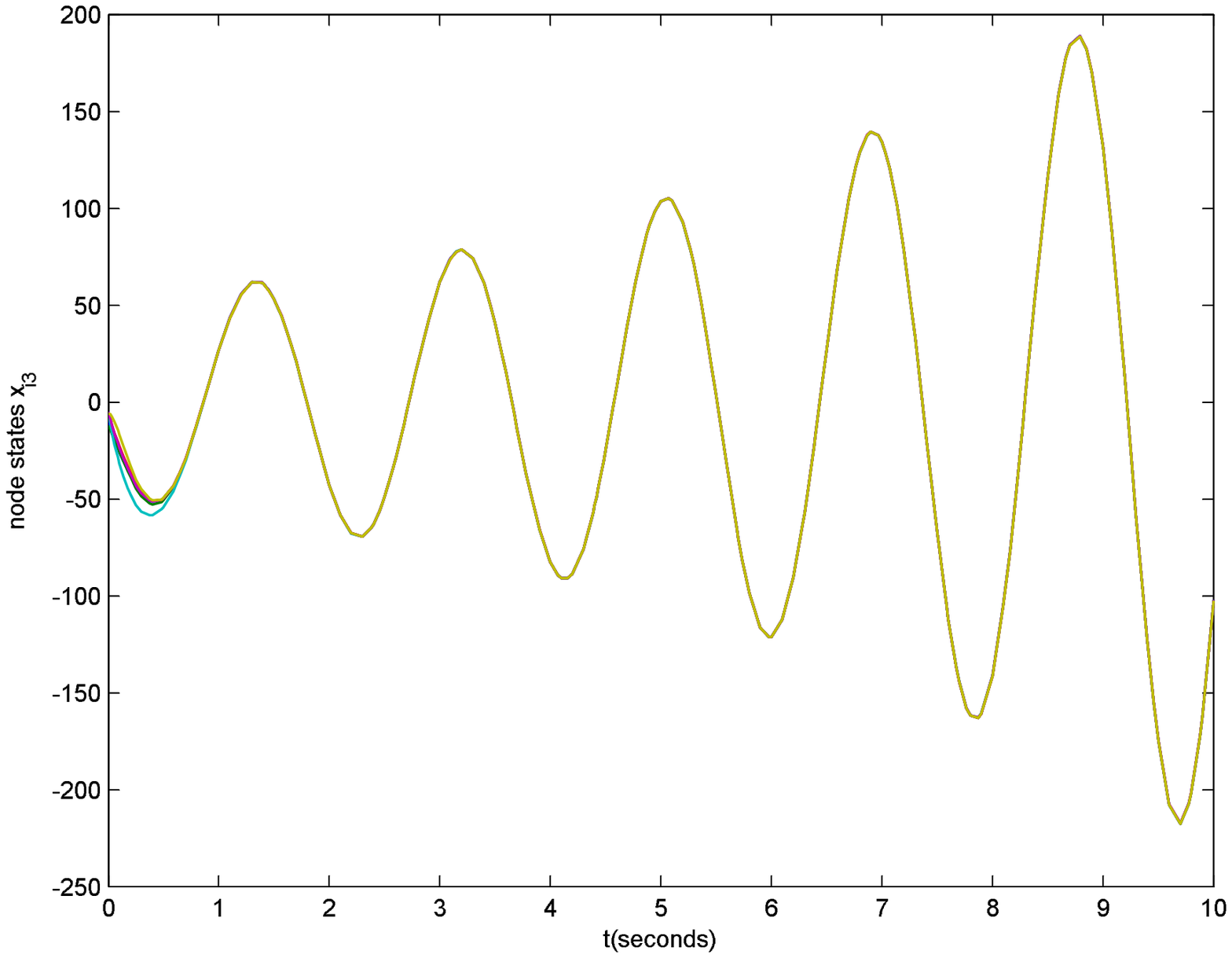}\label{fig:quasistrongly_x3}}
}
\caption{The simulation results under a quasi-strongly connected digraph.}
\label{fig:quasistrongly}
\end{center}
\end{figure}
\section{CONCLUSIONS}
The edge Laplacian of digraph and its related concepts were originally proposed in this paper. Based on these graph-theoretic tools, we developed \ul{a} new systematic framework to study multi-agent system in the context of the edge agreement. To show how the edge Laplacian sheds a new light on the leaderless consensus problem, the technical challenges caused by the unknown but bounded disturbances and the inherently nonlinear dynamics were considered\ul{;} and the classical ISS nonlinear control methods together with the recently developed cyclic-small-gain theorem were successfully implemented to drive multi-agent system to reach robust consensus. Furthermore, the edge-interconnection graph, which plays an important role in the analysis and synthesis of multi-agent networks, was proposed and its intricate relationship with the original graph was discussed. For the quasi-strongly connected case, we also pointed out a reduced order modeling for the edge agreement in terms of the spanning tree subgraph. Based on this observation, by guaranteeing the ISS properties of each subsystem and assigning the appropriate gains for both of the interconnected subsystems to satisfy the small gain condition, the closed-loop multi-agent system could reach robust consensus. Under the edge agreement framework, we believe nonlinear multi-agent systems with more complex factors, such as switching topologies and time-delays, will be well settled.

\ack This work was supported by the National Natural Science Foundation of China under grant 61403406.

\bibliographystyle{plain}                                                               
\bibliography{IJRNCconsensus}

\begin{thebibliography}{10}

\bibitem{bauso2009consensus}
Dario Bauso, Laura Giarr{\'e}, and Raffaele Pesenti.
\newblock Consensus for networks with unknown but bounded disturbances.
\newblock {\em SIAM Journal on Control and Optimization}, 48(3):1756--1770,
  2009.

\bibitem{beard2002coordinated}
Randal~W Beard, Timothy~W McLain, Michael~A Goodrich, and Erik~P Anderson.
\newblock Coordinated target assignment and intercept for unmanned air
  vehicles.
\newblock {\em Robotics and Automation, IEEE Transactions on}, 18(6):911--922,
  2002.

\bibitem{ben2003generalized}
A~Ben-Israel and TNE Greville.
\newblock Generalized inverses: theory and applications, 1974.
\newblock {\em Willey, New York}, 2003.

\bibitem{das2011cooperative}
Abhijit Das and Frank~L Lewis.
\newblock Cooperative adaptive control for synchronization of second-order
  systems with unknown nonlinearities.
\newblock {\em International Journal of Robust and Nonlinear Control},
  21(13):1509--1524, 2011.

\bibitem{godsil2001algebraic}
Christopher~David Godsil, Gordon Royle, and CD~Godsil.
\newblock {\em Algebraic graph theory}, volume 207.
\newblock Springer New York, 2001.

\bibitem{isidori1999nonlinear}
Alberto Isidori.
\newblock {\em Nonlinear Control Systems II}, volume~2.
\newblock Springer, 1999.

\bibitem{jie2012containment}
Mei Jie, Ren Wei, and Ma~Guangfu.
\newblock Containment control for multiple unknown second-order nonlinear
  systems under a directed graph based on neural networks.
\newblock In {\em Control Conference (CCC), 2012 31st Chinese}, pages
  6450--6455. IEEE, 2012.

\bibitem{li2012global}
Zhongkui Li, Xiangdong Liu, Mengyin Fu, and Lihua Xie.
\newblock Global $h_\infty$ consensus of multi-agent systems with lipschitz
  non-linear dynamics.
\newblock {\em Control Theory \& Applications, IET}, 6(13):2041--2048, 2012.

\bibitem{lin2005necessary}
Zhiyun Lin, Bruce Francis, and Manfredi Maggiore.
\newblock Necessary and sufficient graphical conditions for formation control
  of unicycles.
\newblock {\em Automatic Control, IEEE Transactions on}, 50(1):121--127, 2005.

\bibitem{lin2007state}
Zhiyun Lin, Bruce Francis, and Manfredi Maggiore.
\newblock State agreement for continuous-time coupled nonlinear systems.
\newblock {\em SIAM Journal on Control and Optimization}, 46(1):288--307, 2007.

\bibitem{liu2011lyapunov}
Tengfei Liu, David~J Hill, and Zhong-Ping Jiang.
\newblock Lyapunov formulation of iss cyclic-small-gain in continuous-time
  dynamical networks.
\newblock {\em Automatica}, 47(9):2088--2093, 2011.

\bibitem{liu2013distributed}
Tengfei Liu and Zhong-Ping Jiang.
\newblock Distributed formation control of nonholonomic mobile robots without
  global position measurements.
\newblock {\em Automatica}, 49(2):592--600, 2013.

\bibitem{liu2013outputfeedback}
Tengfei Liu and Zhong-Ping Jiang.
\newblock Distributed output-feedback control of nonlinear multi-agent systems.
\newblock {\em Automatic Control, IEEE Transactions on}, 58(11):2912--2917,
  2013.

\bibitem{mei2013distributed}
Jie Mei, Wei Ren, and Guangfu Ma.
\newblock Distributed coordination for second-order multi-agent systems with
  nonlinear dynamics using only relative position measurements.
\newblock {\em Automatica}, 49(5):1419--1427, 2013.

\bibitem{mesbahi2010graph}
Mehran Mesbahi and Magnus Egerstedt.
\newblock {\em Graph theoretic methods in multiagent networks}.
\newblock Princeton University Press, 2010.

\bibitem{meyer2000matrix}
Carl~D Meyer.
\newblock {\em Matrix analysis and applied linear algebra}, volume~2.
\newblock Siam, 2000.

\bibitem{nedic2009distributed}
Angelia Nedic and Asuman Ozdaglar.
\newblock Distributed subgradient methods for multi-agent optimization.
\newblock {\em Automatic Control, IEEE Transactions on}, 54(1):48--61, 2009.

\bibitem{olfati2007distributed}
Reza Olfati-Saber.
\newblock Distributed kalman filtering for sensor networks.
\newblock In {\em Decision and Control, 2007 46th IEEE Conference on}, pages
  5492--5498. IEEE, 2007.

\bibitem{olfati2004consensus}
Reza Olfati-Saber and Richard~M Murray.
\newblock Consensus problems in networks of agents with switching topology and
  time-delays.
\newblock {\em Automatic Control, IEEE Transactions on}, 49(9):1520--1533,
  2004.

\bibitem{ren2009distributed}
Wei Ren.
\newblock Distributed leaderless consensus algorithms for networked
  euler--lagrange systems.
\newblock {\em International Journal of Control}, 82(11):2137--2149, 2009.

\bibitem{ren2005consensus}
Wei Ren, Randal~W Beard, et~al.
\newblock Consensus seeking in multiagent systems under dynamically changing
  interaction topologies.
\newblock {\em IEEE Transactions on Automatic Control}, 50(5):655--661, 2005.

\bibitem{shi2009global}
Guodong Shi and Yiguang Hong.
\newblock Global target aggregation and state agreement of nonlinear
  multi-agent systems with switching topologies.
\newblock {\em Automatica}, 45(5):1165--1175, 2009.

\bibitem{sontag2008input}
Eduardo~D Sontag.
\newblock Input to state stability: Basic concepts and results.
\newblock In {\em Nonlinear and optimal control theory}, pages 163--220.
  Springer, 2008.

\bibitem{tanner2004leader}
Herbert~G Tanner, George~J Pappas, and Vijay Kumar.
\newblock Leader-to-formation stability.
\newblock {\em Robotics and Automation, IEEE Transactions on}, 20(3):443--455,
  2004.

\bibitem{thulasiraman2011graphs}
Krishnaiyan Thulasiraman and Madisetti~NS Swamy.
\newblock {\em Graphs: theory and algorithms}.
\newblock John Wiley \& Sons, 2011.

\bibitem{wang2012second}
Xiangke Wang, Tengfei Liu, and Jiahu Qin.
\newblock Second-order consensus with unknown dynamics via cyclic-small-gain
  method.
\newblock {\em Control Theory \& Applications, IET}, 6(18):2748--2756, 2012.

\bibitem{wang2012formation}
Xiangke Wang, Jiahu Qin, Xun Li, and Zhiqiang Zheng.
\newblock Formation tracking for nonlinear agents with unknown second-order
  locally lipschitz continuous dynamics.
\newblock In {\em Control Conference (CCC), 2012 31st Chinese}, pages
  6112--6117. IEEE, 2012.

\bibitem{wang2014coordination}
Xiangke Wang, Jiahu Qin, and Changbin Yu.
\newblock Iss method for coordination control of nonlinear dynamical agents
  under directed topology.
\newblock {\em IEEE Transactions on Cybernetics}, 14, 2014.

\bibitem{wen2012consensus}
Guanghui Wen, Zhisheng Duan, Zhongkui Li, and Guanrong Chen.
\newblock Consensus and its $l_2$-gain performance of multi-agent systems with
  intermittent information transmissions.
\newblock {\em International Journal of Control}, 85(4):384--396, 2012.

\bibitem{yu2010second}
Wenwu Yu, Guanrong Chen, Ming Cao, and J{\"u}rgen Kurths.
\newblock Second-order consensus for multiagent systems with directed
  topologies and nonlinear dynamics.
\newblock {\em Systems, Man, and Cybernetics, Part B: Cybernetics, IEEE
  Transactions on}, 40(3):881--891, 2010.

\bibitem{zelazo2011edge}
Daniel Zelazo and Mehran Mesbahi.
\newblock Edge agreement: Graph-theoretic performance bounds and passivity
  analysis.
\newblock {\em Automatic Control, IEEE Transactions on}, 56(3):544--555, 2011.

\bibitem{zelazo2011graph}
Daniel Zelazo and Mehran Mesbahi.
\newblock Graph-theoretic analysis and synthesis of relative sensing networks.
\newblock {\em Automatic Control, IEEE Transactions on}, 56(5):971--982, 2011.

\bibitem{zelazo2007agreement}
Daniel Zelazo, Amirreza Rahmani, and Mehran Mesbahi.
\newblock Agreement via the edge laplacian.
\newblock In {\em Decision and Control, 2007 46th IEEE Conference on}, pages
  2309--2314. IEEE, 2007.

\bibitem{zelazo2013performance}
Daniel Zelazo, Simone Schuler, and Frank Allg{\"o}wer.
\newblock Performance and design of cycles in consensus networks.
\newblock {\em Systems \& Control Letters}, 62(1):85--96, 2013.

\bibitem{zeng2014nonlinear}
Zhiwen Zeng, Xiangke Wang, and Zhiqiang Zheng.
\newblock Nonlinear consensus under directed graph via the edge laplacian.
\newblock In {\em Control and Decision Conference (2014 CCDC), The 26th
  Chinese}, pages 881--886. IEEE, 2014.

\bibitem{zhu2010leader}
Wei Zhu and Daizhan Cheng.
\newblock Leader-following consensus of second-order agents with multiple
  time-varying delays.
\newblock {\em Automatica}, 46(12):1994--1999, 2010.

\end{thebibliography}

\end{document}